\documentclass[11pt, a4paper]{article} 

\usepackage{a4}
\usepackage{geometry}

\usepackage{latexsym}
\usepackage{amssymb}
\usepackage{amsfonts}
\usepackage{mathrsfs}
\usepackage{amsmath}
\usepackage{amsthm}
 \usepackage[colorlinks=true,linkcolor=black,citecolor=black]{hyperref}
\usepackage[pdftex]{graphicx}

\usepackage{geometry}

\newcommand{\bfi}{\begin{fig}}
\newcommand{\efi}{\end{fig}}

\newcommand{\btab}{\begin{tab}}
\newcommand{\etab}{\end{tab}}

\newcommand{\barr}{\begin{array}}
\newcommand{\earr}{\end{array}}

\newcommand{\beqq}{\begin{equation}}
\newcommand{\eeqq}{\end{equation}}

\newcommand{\beao}{\begin{eqnarray*}}
\newcommand{\eeao}{\end{eqnarray*}\noindent}

\newcommand{\beam}{\begin{eqnarray}}
\newcommand{\eeam}{\end{eqnarray}\noindent}

\newcommand{\bdis}{\begin{displaymath}}
\newcommand{\edis}{\end{displaymath}\noindent}

\def\bbr{{\Bbb R}}   
\def\bbn{{\Bbb N}}

\def\bbl{{\Bbb L}}

\newcommand{\eps}{{\varepsilon}}

\newcommand{\vp}{{\varphi}}

\newcommand{\wh}{\widehat}
\newcommand{\wt}{\widetilde}

\newcommand{\auf}{[\![}

\newcommand{\mal}{\stackrel{\mbox{\tiny$\bullet$}}{}}

\newtheorem{Satz}{Theorem}[section]

\newtheorem{Proposition}[Satz]{Proposition}
\newtheorem{Lemma}[Satz]{Lemma}
\newtheorem{Definition}[Satz]{Definition}

\newtheorem{Bemerkung}[Satz]{Remark}

\usepackage{stmaryrd}
\usepackage{amsmath,amsthm}
\usepackage{amssymb}
\usepackage[numbers,square]{natbib}

\newtheorem{thmD}[Satz]{Definition}
\newtheorem{thmB}[Satz]{Remark}
\newtheorem{thmP}[Satz]{Proposition}
\newtheorem{thmS}[Satz]{Theorem}
\newtheorem{thmL}[Satz]{Lemma}

\usepackage{graphicx}
\usepackage{subfigure}
\usepackage{framed}

\begin{document}
\bibliographystyle{apalike}
\numberwithin{equation}{section}

\title{Modeling capital gains taxes for trading strategies of infinite variation\thanks{The authors thank Christoph Czichowsky and Teemu Pennanen for fruitful discussions and an anonymous
referee for valuable comments.}} 
\author{Christoph K\"uhn\thanks{Institut f\"ur Mathematik, Goethe-Universit\"at Frankfurt, D-60054 Frankfurt a.M., Germany, e-mail: \{ckuehn, ulbricht\}@math.uni-frankfurt.de} 
\and Bj\"orn Ulbricht\footnotemark[2] }
\date{}
\maketitle
\sloppy

\begin{abstract}
In this article, we show that the payment flow of a linear tax on trading gains from a security with a semimartingale price process
can be constructed for all c\`agl\`ad and adapted trading strategies. It is characterized as the unique continuous extension of the tax payments for elementary strategies
w.r.t. the convergence ``uniformly in probability''. 
In this framework, we prove that under quite mild assumptions dividend payoffs have almost surely a negative effect on investor's after-tax wealth
if the riskless interest rate is always positive. 
In addition, we give an example for tax-efficient strategies for which the tax payment flow can be computed explicitly.
\end{abstract}

\begin{tabbing}


\end{tabbing}
\section{Introduction}
\setcounter{equation}{0}

%

In this article, we want to answer the following question. Can tax payments on capital gains be modeled for continuous time trading strategies of the kind they generally
appear in mathematical finance~? Most of these strategies possess infinite variation, as, e.g., the optimal stock position in the Merton problem or the replicating portfolio 
of an option in the Black Scholes model. A straight forward construction of the tax payment flow, analogous to time-discrete models, would be based both on accumulated purchases and 
accumulated sales of assets. But, of course, these quantities explode if strategies are of infinite variation. 

For simplicity, we consider a {\em linear} taxing rule with tax rate~$\alpha\in (0,1)$, i.e., if an asset 
with stochastic price process~$S$ is purchased at time $t_1$ and sold at time $t_2$, the
trading gains $S_{t_2}-S_{t_1}$ are taxed at $\alpha (S_{t_2}-S_{t_1})$. 
Negative tax payments for losses, so-called tax 
credits, can be interpreted as a refund of former tax payments or a deduction against future tax payments. 

An important feature of the tax code is the fact that trading gains are not taxed before the asset 
is liquidated, i.e., the gain is realized. Thus, the investor can influence the timing of the tax payments, namely she holds a {\em deferral option}. 
Possible dividend payments are taxed immediately. 
A crucial observation is the following. If the investor buys, e.g., 100 General Motors stocks at time $t_1$, another 100 at time~$t_2$, and sells 100 at 
time~$t_3$, it matters {\em which} of the stocks she sells, as in general $\alpha\cdot 100 (S_{t_3}-S_{t_2})\not=\alpha\cdot 100 (S_{t_3}-S_{t_1})$. 
When the portfolio is liquidated at some date~$t_4$
the difference of the accumulated tax payments disappears because $\alpha\cdot 100 (S_{t_3}-S_{t_2}) + \alpha\cdot 100 (S_{t_4}-S_{t_1}) 
= \alpha\cdot 100 (S_{t_3}-S_{t_1}) + \alpha\cdot 100 (S_{t_4}-S_{t_2})$.
But, the order of sales still matters for discounted payments if the riskless interest rate does not vanish.
In the case of a positive riskless interest rate, it is more favorable to realize smaller trading gains first.
Moreover, if the stock falls below its purchasing price, it is worthwhile to sell it in order to realize the trading loss and rebuy 
it immediately, which is called a {\em wash sale}. These facts were already observed in \citet{dyb1}, see Properties~1 and 2 on page 6. 
For a rigorous proof of these seemingly obvious statements 
considering arbitrary dynamic trading strategies, see Appendix~\ref{25.9.2014.1} of the current paper. 
For investors, wash sales are a method to claim a capital loss without actually changing their position.
The regulation described above that leaves it up to the taxpayer to choose which trading gain to realize first when a stock position is reduced is called the {\em exact tax basis}.
An example is the U.S. tax law that allows investors to use a separate tax 
basis for each security. But, the U.S. tax law disallows loss deductions if the
same stock is repurchased within thirty days. However, this regulation can
easily be bypassed by purchasing a similar stock. 
There are also other tax codes, specifing the basis to which the price 
of a security has to be compared in order to evaluate the capital gains 
(or losses). In some countries, the 
basis is the average purchase price of all stocks of the same firm (e.g., 
in Canada) or the price of the stock which was  
bought first (``first-in-first-out'', a procedure followed, e.g., in Germany). Of course, the exact tax basis offers the
investor the maximal possible flexibility to make use of her tax-timing option. Economically, the exact tax basis seems to be the most reasonable one because 
highly correlated stocks of different firms are anyhow considered separately. 

Although in practice capital gains taxes may be the most relevant market friction, 
there is only little literature on capital gains taxes in advanced continuous time models. 
Ben Tahar, Soner, and Touzi~\cite{bentaharsonertouzi10,bentaharsonertouzi07} solve the Merton problem 
with proportional transaction costs and a tax based on the average of past purchasing prices. 
This approach has the advantage that the optimization problem is Markov with the one-dimensional tax basis 
as additional state variable.
Cadenillas and Pliska~\cite{cadenillas1999optimal} and Buescu, Cadenillas, and Pliska~\cite{buescu2007note} maximize the long-run growth rate of investor's wealth 
in a model with taxes and transaction costs. Here, after each portfolio regrouping, the investor has to pay capital gains taxes for her total portfolio.
Jouini, Koehl, and Touzi~\cite{jouni2,jouni1} consider the first-in-first-out priority rule with one nondecreasing asset price, but with a quite general tax code, 
and derive first-order conditions for the optimal consumption
problem. The problem consists of injecting cash from the income stream into the single asset and withdrawing it for consumption. Consequently, all admissible strategies are of finite
variation. \citet{dyb1} and \citet{dem1} model the exact tax basis, as in the current article, but in discrete time and relate the portfolio optimization problem to nonlinear programming.
 
Whereas in models with proportional transaction costs it is quite obvious that strategies of exploding variation lead to exploding costs and thus
to an immediate ruin for sure, capital gains taxes do not explode. Namely, taxes are not triggered by portfolio regroupings alone if there are no price changes. 
In addition, even if the investment strategy forces that gains from upward movements of the stock are realized, there is to some extent an offset by losses
due to tax credits. 
On the other hand, a straightforward generalization of the model by \cite{dyb1, dem1}
to continuous time is only available for finite variation strategies -- as not only the number of shares held in the portfolio enters in the self-financing condition, 
but it is based on both purchases 
and sells. In this article, we show how tax payments can nevertheless be constructed under the condition that stocks are semimartingales.\\ 

One application is to compare different dividend policies. 
As dividend payoffs, in contrast to (unrealized) book profits, have to be taxed immediately, capital gains taxes are also relevant for dividend policies. 
Among economists, there have been extensive discussions about optimal 
dividend policies. In the famous article by \citet{miller1961dividend}, their effect on the current stock price is considered, and their irrelevance 
for the firm valuation is shown in perfect markets (i.e., without taxes). A question arising from \cite{miller1961dividend} is: ``Why do firms pay dividends?''. 
The so-called {\em dividend puzzle}, at first appearing 
in \citet{black1976dividend}, states that there are no rational reasons for a firm to pay dividends. 
\citet{bernheim1990tax} solves this puzzle considering a model (with taxes) in which firms attempt to signal profitability by distributing cash to shareholders.
For a survey on these general, but mainly less formal, discussions on dividend policies we refer to the book of \citet{RePEc}.

The current article does {\em not} make any contribution to the solution of the dividend puzzle. Instead, we establish precise conditions under which the widely held view 
that dividends have a negative impact on investors' after-tax wealths (cf., e.g., \cite{black1976dividend}) can be proven in a model that allows for dynamic trading. 
If investment opportunities were restricted to a single asset with increasing price process, this property would be quite obvious. 
Indeed, let $r_t>0$ be the growth rate of the asset. By strict convexity of the exponential function, one has
\beam\label{12.9.2013.1}
1+ (1-\alpha)\left(\exp\left(\int_0^t r_s\,ds\right)-1\right)  > \exp\left((1-\alpha)\int_0^t r_s\,ds\right).
\eeam
The LHS of (\ref{12.9.2013.1}) can be interpreted as the value of a portfolio with initial capital~$1$ when capital gains are taxed at time~$t$ with factor~$\alpha$.  
The RHS corresponds to the same situation, but capital gains are already taxed at the time they occur. This tax regulation takes effect if the asset always has price~1 but pays out 
the continuous dividend~$r_t\,dt$ (the after-tax dividend~$(1-\alpha)r_t\,dt$ is then reinvested in the asset). 
However, considering dynamic portfolio regroupings and asset price processes that are not increasing with probability~$1$, a proof of the conjecture that the effect of dividends is always
negative, is, even in discrete time, much trickier than (\ref{12.9.2013.1}). 
We give a proof of this assertion in the continuous time framework provided in this article. 

Finally, to demonstrate the tractability of the model, we give an example for tax-efficient dynamic trading strategies for which the tax payment flow can be computed explicitly and is easy
to interpret.

The article is organized as follows. In Section~\ref{main_results}, we present the model and the first main result, Theorem~\ref{main_theorem}, showing how to construct tax payment 
processes for adapted, left-continuous trading strategies. The construction is based on automatic wash sales and the rule to sell shares with 
shorter residence time first. The optimality
of this procedure is proven in Appendix~\ref{25.9.2014.1} for the discrete time model of Dybvig and Koo~\cite{dyb1}.
In Section~\ref{section_book_profits}, basic properties of the book profits of a portfolio are discussed. 
They are used in the proof of Theorem~\ref{main_theorem} in Section~\ref{section_proofmain}. In Section~\ref{section_self}, the self-financing condition of the model is introduced.
In Section~\ref{section_comparison}, the second main result, Theorem~\ref{propV}, showing that the investor is always better off in a model with a stock which does not pay dividends
is stated and proved. 
Section~\ref{9.10.2014.1} is about tax-efficient strategies, and Section~\ref{2.8.2013.1} gives examples that show the necessity of some assumptions.

\section{Construction of the tax payment process}\label{main_results}

Throughout the article, we fix a terminal time $T\in\mathbb{R}_+$ and a filtered probability space~$(\Omega,\mathcal{F},(\mathcal{F}_t)_{t\in[0,T]},P)$ satisfying the usual conditions. 
Denote by $\mathcal{O}$\index{$\mathcal{O}$} (resp. by $\mathcal{P}$\index{$\mathcal{P}$}) the optional $\sigma$-algebra 
(resp. the predictable $\sigma$-algebra) on $\Omega\times[0,T]$. For optional processes~$X, X^n, n\in\mathbb{N}$, we write $X^n\stackrel{\textrm{up}}{\to} X$ iff $X^n$ converges 
uniformly in probability to $X$, i.e., $\sup_{t\in[0,T]}|X^n_t-X_t|$ converges to $0$ in probability. 
Equality of processes is understood up to evanescence.
A process~$X$ is called l\`agl\`ad iff all paths possess finite left and right limits (but they can have double jumps).
We set $\Delta^+ X:=X_+-X$ and $\Delta X:=\Delta^- X:=X-X_-$, where $X_{t+}:=\lim_{s\downarrow t} X_s$ and $X_{t-}:=\lim_{s\uparrow t} X_s$.
For a random variable $Y$, we set $Y^+:=\max(Y,0)$ and $Y^-:=\max(-Y,0)$. 

For an investor trading in finitely many different stocks, the total tax payment is just the sum of the tax payments considering only gains from one type of stock. Thus, 
it is sufficient to consider only one risky asset (sometimes called stock). 
Its price process is given by the semimartingale~$(S_t)_{t\in[0,T]}$ (thus the paths are c\`adl\`ag).
The stock pays out nonnegative dividends.
Accumulated dividends per share are modeled by the nondecreasing adapted c\`adl\`ag process $(D_t)_{t\in[0,T]}$. 
All capital gains (positive or negative) are taxed with the rate $\alpha\in(0,1)$.
But, whereas dividends are taxed immediately, trading gains arising from stock price movements are not taxed before they are realized. 
Denote by $\mathbb{L}$ the set of all left-continuous adapted processes possessing finite right limits. The investor's strategy is the number of identical stocks she holds,
and it is modeled by some $\varphi\in\mathbb{L}$ with $\varphi_0=0$ and $\varphi\ge 0$. Short-selling is forbidden as otherwise the investor can hold {\em one} long and {\em one} 
short position of the same stock at the same time, and this can lead to an arbitrage opportunity under a linear 
tax rule and a positive riskless interest rate (losses are immediately realized, and the corresponding gains are deferred, cf. Constantinides~\cite{constantinides.1983}). 
The assumption that $\varphi_0=0$ is solely for notational 
convenience (cf. (\ref{25.7.2013.3})). It does not rule out that the investor starts with a bulk trade $\varphi_{0+}>0$.
%

\begin{Bemerkung}
In general, the tax payment flow cannot be derived from the process~$\vp$ alone as payments depend on {\em which} shares the investor sells when $\vp$ is reduced
and on the occurrence of wash sales
that do not enter in $\vp$. Given some $\vp$, we work with a special procedure that dictates which of the shares to sell.  
In Appendix~\ref{25.9.2014.1}, for a nonnegative interest rate, the pathwise optimality of this procedure is proven 
in the discrete time model of Dybvig and Koo~\cite{dyb1} where arbitrary shares can be sold. 
We use that a payment obligation in the future is prefered to a payment obligation today. With this intuition in mind, the constructions in the 
current section are well-founded, but there are also good reasons to read Appendix~\ref{25.9.2014.1} first.
\end{Bemerkung}

\begin{Bemerkung}
It is important to note that the pathwise optimality of wash sales in the model of Dybvig and Koo that motivates our model, see also Theorem~\ref{theorem:main111},
is based on the absence of transaction
costs. With proportional transaction costs, there would be a trade-off between the aim to realize losses immediately and the aim to avoid transaction costs.
The (non-)optimality of wash sales would depend on the size of book losses and transaction costs, but also on (the probability law of) future asset price movements (e.g.,
if there is a reason to liquidate the asset anyhow shortly afterwards, a wash sale is less profitable). 
Thus, in the presence of transaction costs, one cannot reduce the strategy of \cite{dyb1} independently of investor's beliefs and preferences to a one-dimensional 
predictable process~$(\vp_t)_{t\in[0,T]}$
that only specifies the total number of shares in the portfolio.
This means that the tractability of our model is essentially based on the absence of transaction costs. 
Consequently, transaction costs cannot be used to rule out trading strategies of infinite variation.
\end{Bemerkung}

To construct the tax payment process, several mathematical objects have to be introduced. For every $t$, we sort the $\vp_t$ stocks by the time spending in the portfolio and 
label them by $x$: the larger $x$ the longer the residence time in the portfolio. 
We follow the above-mentioned procedure:
\beam\label{15.8.2014.1}
\mbox{\em ``latest purchased stocks are sold first''.} 
\eeam
With this procedure, the purchasing time of the $x$th stock is defined by 
\beam\label{13.8.2014.1}
\tau_{t,x}:=\left\{ 
\begin{array}{cl}
\sup M_{t,x}  & \textrm{if }M_{t,x}\not=\emptyset\\
 t & \textrm{otherwise}
\end{array}
\right.,\quad t\in [0,T], x\in\bbr_+,
\eeam
where $M_{t,x} := \{u\in\mathbb{R}_+\ |\  (u\le t\ \mbox{and}\ x-\varphi_t + \varphi_u \le 0)\ \mbox{or}\ (u< t\ \mbox{and}\ x-\varphi_t + \varphi_{u+} \le 0)\}$. 
By $\varphi_0=0$ and $\varphi\ge 0$, one has that 
\beam\label{19.9.2013.1}
M_{t,x}=\emptyset \Leftrightarrow x > \varphi_t
\eeam
and thus
\beam\label{25.7.2013.3}
\tau_{t,x} = 1_{(x\le \varphi_t)} \sup M_{t,x} + 1_{(x > \varphi_t)} t.
\eeam
The construction is illustrated in Figure~\ref{12.8.2014.1}.
\begin{figure}[ht]
    \subfigure{\includegraphics[width=0.98\textwidth, height=190px]{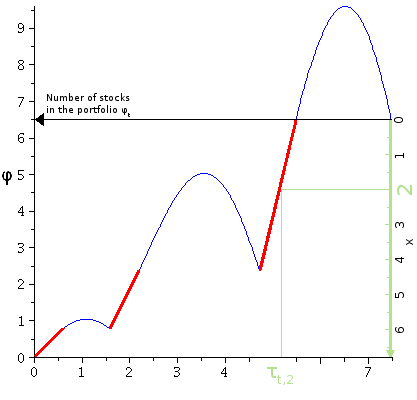}}
\caption{On the ordinate, the stocks that are in the portfolio at time~$t$ are sorted by descending label~$x$ (see the green axis). 
$\tau_{t,x}$, the purchasing time of stock~$x$, is the last time~$u$ before 
$t$ with $\vp_u=\vp_t-x$ (see the case $x=2$). 
The pieces that are marked in red symbolize the stocks (and their purchasing times) which are still in the portfolio at time~$t$.
%
If the position is reduced, stocks with lower residence time in the portfolio are sold first. 
}\label{12.8.2014.1}
\end{figure} 
Next, an automatic loss realization is modeled.
The trading gain of piece~$x$ is decomposed into
\beam\label{13.8.2014.2}
S_t - S_{\tau_{t,x}} = \underbrace{\inf_{\tau_{t,x}\leq u\leq t}S_u - S_{\tau_{t,x}}}_{\mbox{realized losses by wash sales}} 
+ \underbrace{S_t - \inf_{\tau_{t,x}\leq u\leq t}S_u}_{\mbox{unrealized book profits}}.
\eeam
This is motivated as follows: if a stock falls below its purchasing price, it is sold and rebought in order to declare a loss. 
Then, in the continuous time limit, the realized loss is the first summand on the
RHS of (\ref{13.8.2014.2}). The residual second summand are the unrealized book profits. 
\begin{thmD}[Book profits]\label{bookprofits}
Let $\varphi\in\mathbb{L}$ with $\varphi_0=0$ and $\varphi\ge 0$. The mapping $F:\Omega\times[0,T]\times\mathbb{R}_+\to\mathbb{R}_+$ with
 \begin{align}
    F_\omega(t,x):=S_t(\omega)-\inf_{\tau_{t,x}(\omega)\leq u\leq t}S_u(\omega),\label{ftx1}
 \end{align}
where $\tau_{t,x}$ is defined in (\ref{13.8.2014.1}), is called the {\em book profit function}.
\end{thmD}
A book profit is a gain that is demonstrated on paper, but not actually real yet.
By the wash sales and the fact that a newly bought share starts with book profit zero, a share with a longer stay in the portfolio
possesses a higher (or equal) book profit, i.e.,  $x\mapsto F_{\omega}(t,x)$ is nondecreasing.

Note that wash sales neither enter into the strategy~$\vp$ (implying that these transactions have no impact on the trading gains) 
nor in the purchasing times~$\tau_{t,x}$. The latter means that $\tau_{t,x}$ is the time at which the share possessing at time $t$ with label~$x$ is bought and kept in 
the portfolio afterwards at least up to time~$t$, {\em apart from later rebuys caused by wash sales}.

\begin{Bemerkung}
The book profit function~(\ref{ftx1}) that depends on the paths of the stock price and the total number of shares turns out to be the key object to construct tax payments 
for strategies of infinite variation and to find out tax-efficient strategies. 
\end{Bemerkung}

\begin{thmP}\label{properties_F}
$F(t,x)$ and $\tau_{t,x}$ fulfill the following properties: 
\begin{itemize}
 \item[(i)] The mapping $x\mapsto\tau_{t,x}$ is nonincreasing on $[0,\varphi_t]$.
 \item[(ii)] $F(t,x)=0$ for $x>\varphi_t$.
 \item[(iii)] $x\mapsto F(t,x)$ is nondecreasing on $[0,\varphi_t]$.
 \item[(iv)] $x\mapsto F(t,x)$ is left-continuous.
 \item[(v)] If $\varphi$ is an elementary strategy, then $\lim_{s\downarrow t}F(s,x)$ exists for all $t,x$.
\end{itemize}
\end{thmP}
The proof can be found at the beginning of Section~\ref{section_book_profits}. Of course, $F(t,x)$ is only used for $x\le \varphi_t$.
Possible states and developments of $F$ over time can be seen in Fig.~\ref{bild1}.

\begin{figure}[ht]
\subfigure[$S_{1}=103,~\varphi_{1}=9,~\varphi_{2}-\varphi_{1}=1$]{\includegraphics[width=0.49\textwidth, height=120px]{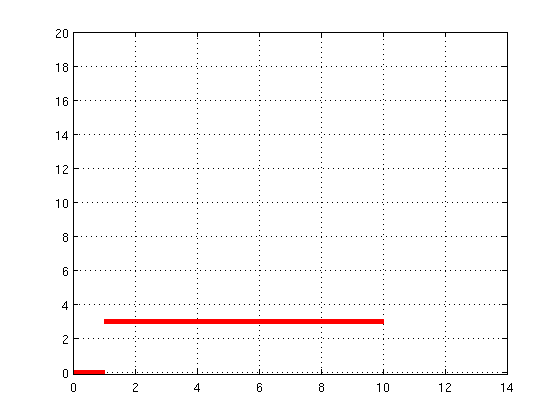}} 
    \subfigure[$S_{2}=104,~\varphi_{2}=10,~\varphi_{3}-\varphi_{2}=4$]{\includegraphics[width=0.49\textwidth, height=120px]{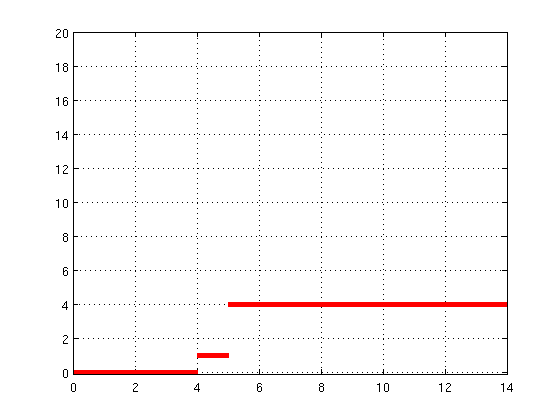}}\\
    \subfigure[$S_{3}=105,~\varphi_{3}=14,~\varphi_{4}-\varphi_{3}=-4$]{\includegraphics[width=0.49\textwidth, height=120px]{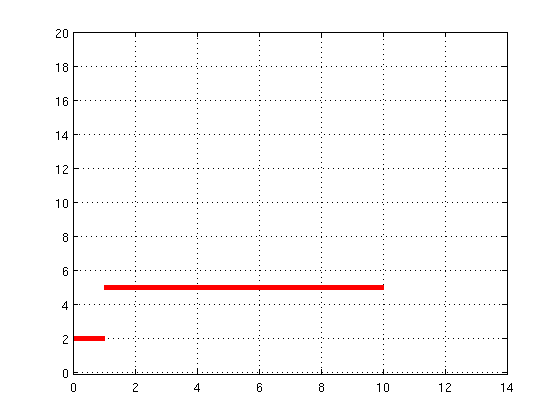}}
    \subfigure[$S_{4}=102,~\varphi_{4}=10,~\varphi_{5}-\varphi_{4}=0$]{\includegraphics[width=0.49\textwidth, height=120px]{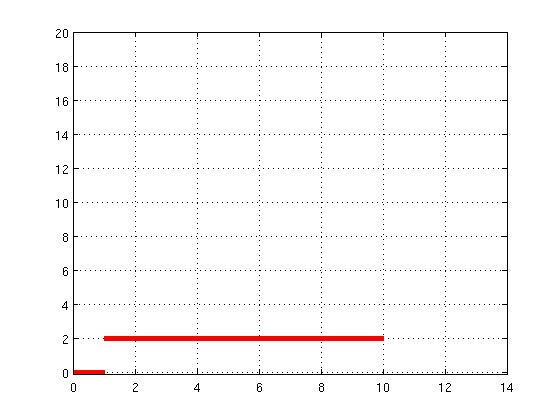}}     
\caption{An example how $x\mapsto F(t,x)$ can evolve in a 4-period model, i.e., $t\in\{0,1,2,3,4\}$. The stock price is given by $S=(S_0,\ldots,S_4)=(100,103,104,105,102)$,
and the investor chooses the strategy~$\varphi=(\varphi_1,\ldots,\varphi_5)=(9,10,14,10,10)$, following the standard notation in discrete time, i.e., $\vp_1$ shares are purchased at price
$S_0$ etc. On the abscissa there are the shares ordered by their book profits and on the ordinate the book profits~$F(t+,x)$, i.e., after the portfolio regrouping at time~$t$.
Observe that at time~$t=4$, i.e., in the fourth picture, one share (at the very left) is sold and bought back to realize a loss of one monetary 
unit (wash sale).}\label{bild1}
\end{figure} 

\begin{Bemerkung}
To ensure that the function~$x\mapsto F(t,x)$ is left-continuous, besides $\varphi_u$, also $\varphi_{u+}$ has to be considered in the definition of $M_{t,x}$.
It is convenient that $x\mapsto F(t,x)$ does not possess double jumps, but for the following construction of the tax payment process the values of $F$ at the (countably many) points of 
discontinuity do not matter. $F_\omega(t,\cdot)\mid_{(0,\varphi_t(\omega)]}$ can also be seen as the left-continuous inverse of the distribution function of the book profits over all shares that are in the 
portfolio at time $t$ (here, ``distribution function'' means the number of shares with book profits lower than or equal to a given bound).
\end{Bemerkung}

Whereas the book profit function in (\ref{ftx1}) is directly defined for all $\varphi\in\mathbb{L}$, it turns out that a straight forward construction 
of the tax payment process, analogous to time-discrete models, would be based on both the accumulated purchases and the accumulated sales (this is as both effects are quite different). 
Thus, in a first step, the tax payments are only defined for elementary 
strategies. Then, in Theorem~\ref{main_theorem} we show that it can be extended to all left-continuous adapted processes. However, this extension is not obvious and relies, 
among other things, on the
assumption that $S$ is a semimartingale (see Remark~\ref{non_semimartingale}). With the help of (\ref{ftx1}), a process $\Pi$ can be defined which reflects the accumulated tax payments 
up to time~$t$. 

\begin{thmD}[Accumulated tax payments for elementary strategies]\label{Pi_elementar}
Let $\varphi$ be a nonnegative elementary strategy s.t. $\varphi=\sum_{i=1}^{k}H_{i-1} 1_{\rrbracket \kappa_{i-1},\kappa_i\rrbracket},$ where 
$0=\kappa_0\le\kappa_1\le \ldots\le \kappa_k=T$ are stopping times and $H_{i-1}$ is $\mathcal{F}_{\kappa_{i-1}}-$measurable. 
Let $\tau$ and $F$ be as in Definition \ref{bookprofits}. Then, 
\begin{align}\label{16.10.2014.1}
 \Pi_t(\varphi):=&\alpha\sum_{i=1}^k 1_{(\kappa_{i-1}<t)}
 \int_0^{(H_{i-1}-H_{i-2})^-}F(\kappa_{i-1},x)\,dx\nonumber\\
 &+\alpha\sum_{i=1}^k 1_{(\kappa_{i-1}<t)}\int_0^{\varphi_t}\left(F(\kappa_{i-1}+,x)+\inf_{\kappa_{i-1}\leq u\leq t\wedge \kappa_i}(S_u-S_{\kappa_{i-1}})\right)\wedge 0 \,dx 
 +\alpha \int_0^t \varphi_u dD_u,
\end{align}
where $H_{-1}:=0$, is the {\em tax payment process} of the elementary strategy $\varphi$ (The limit $F(\kappa_{i-1}+,x):=\lim_{s\downarrow \kappa_{i-1}}F(s,x)$ exists by 
Proposition~\ref{properties_F}(v)).
\end{thmD}
$\Pi$ is obviously well-defined, i.e., it does not depend on the representation of $\varphi$.
\begin{thmB}\label{24.9.2013.1}
Let us explain the three components of $\Pi_t(\vp)$. $\alpha\sum_{i=1}^k 1_{(\kappa_{i-1}<t)}
 \int_0^{(H_{i-1}-H_{i-2})^-}F(\kappa_{i-1},x)\,dx$ are the tax payments that are triggered by selling stocks in order to follow the strategy~$\vp$. 
A downward jump of $\vp$ forces the investor to realize book profits. She takes the shares with the smallest label~$x$, which is in line with (\ref{15.8.2014.1}) and (\ref{13.8.2014.1}).
As $x\mapsto F(s,x)$ is nondecreasing, the sold shares possess the lowest book 
profits of all shares in the portfolio. By $F\ge 0$, this term is nonnegative.

$\alpha\sum_{i=1}^k 1_{(\kappa_{i-1}<t)}\int_0^{\varphi_t}\left(F(\kappa_{i-1}+,x)+\inf_{\kappa_{i-1}\leq u\leq t\wedge \kappa_i}(S_u-S_{\kappa_{i-1}})\right)\wedge 0 \,dx$
is always less or equal to zero. The $i$th summand models the tax credits due to realized losses by wash sales between the trading times $\kappa_{i-1}$ and $\kappa_i$. 
This equals minus the local time of $S$ at different levels (in the sense of Asmussen~\cite{asmussen.2003}, page 251). 
Namely, the book profit of piece~$x$ is the solution of a Skorokhod problem 
started at $F(\kappa_{i-1}+,x)$ in which the stock price movements are reflected at $0$ (however, this interpretation is only valid in between portfolio regroupings).
The local time we consider has jumps iff downward price jumps dominate previous book profits. It is different from the semimartingale local time, see  
(5.47) in Jacod~\cite{jacod.1979} for a definition. But, for $S$ being a continuous local martingale, the semimartingale local time of the reflected stock price 
is twice the local time in \cite{asmussen.2003}, see the appendix of Yor~\cite{yor.1979}.

$\alpha \int_0^t \varphi_u dD_u$ are taxes on dividends, which have to be paid immediately.
\end{thmB}

\begin{Bemerkung}
Given an elementary process~$\vp$ modeling the total number of shares in the portfolio, 
$\Pi_t(\vp)$ are the minimal accumulated tax payments up to time~$t$. This statement follows from Theorem~\ref{theorem:main111} together with Subsection~\ref{3.10.2014.1}.
\end{Bemerkung}

(\ref{16.10.2014.1}) can generally not be formulated for strategies of infinite variation.

\begin{Bemerkung}
It is quite natural that the tax payment process has double jumps. Namely, the stock price is right-continuous whereas the strategy is left-continuous, 
and the tax payments are triggered both by downward jumps of the stock (through wash sales) and by sales of stocks following the strategy $\varphi$.  
\end{Bemerkung}

\begin{thmS}\label{main_theorem}
Let $\varphi\in\mathbb{L}$ and $(\varphi^n)_{n\in\mathbb{N}}$ be a sequence of elementary strategies 
with $\varphi^n_0=0,\ \varphi^n\ge 0$, and $\varphi^n\stackrel{\textrm{up}}{\rightarrow}\varphi$. Then, the accumulated tax payments~$\Pi^n$ for $\varphi^n$ 
(as defined in Definition~\ref{Pi_elementar}) are optional processes with l\`agl\`ad paths. 
In addition, there exists an optional process~$\Pi$ possessing almost surely l\`agl\`ad paths such that 
$\Pi^n\stackrel{\textrm{up}}{\rightarrow}\Pi$. Different choices of up-approximating sequences of $\varphi$ lead to the same $\Pi$ up to evanescence.\vspace{0.3cm}\\
\indent Consequently, the mapping $\varphi\mapsto\Pi(\varphi)$ from Definition~\ref{Pi_elementar} possesses an up to evanescence unique extension
\beao
\{\varphi\in\mathbb{L}\ |\ \varphi_0=0,\ \varphi\ge 0\}\to\{X:\Omega\times[0,T]\to\mathbb{R}\ |\ X\ \mbox{is optional and l\`agl\`ad}\}
\eeao
which is continuous w.r.t. the convergence uniformly in probability. The extension, also called $\Pi$, possesses the jumps
\beam
\Delta   \Pi_t= \alpha\int_0^{\varphi_t} \left(\limsup_{s<t, s\to t} F(s,x)+\Delta S_t\right)\wedge 0\,dx +\alpha\varphi_t\Delta D_t
 \quad\mbox{and}\label{jumpPi-}
\eeam
\beam\label{jumpPi+}
\Delta^+ \Pi_t= \alpha\int_0^{(\Delta^+\varphi_t)^-}F(t,x)dx. 
\eeam
\end{thmS}
Note that any $\varphi\in\mathbb{L}$ with $\varphi\ge 0$ can be approximated uniformly in probability by a sequence of nonnegative elementary strategies 
(see, e.g., Theorem~II.10 in \cite{protter1}).
\begin{Definition}\label{6.8.2013.1}
For $\varphi\in\mathbb{L}$ with $\varphi\ge 0$, the tax payment process~$\Pi(\varphi)$ is defined as the limit process in Theorem~\ref{main_theorem}.
\end{Definition}

\begin{Proposition}\label{30.1.2015.1}
The accumulated tax payments are subadditive and positively homogeneous in the trading strategy, i.e., $\Pi(\vp^1+\vp^2)\le \Pi(\vp^1) + 
\Pi(\vp^2)$ and $\Pi(\lambda \vp^1)=\lambda\Pi(\vp^1)$ up to evanescence for all $\vp^1,\vp^2\in\bbl$ with $\vp^1,\vp^2\ge 0$ and $\lambda\in\bbr_+$. 
$\Pi$ is in general not additive. 
\end{Proposition}
The proposition is proven in Subsection~\ref{1.2.1015.1}.

\section{Properties of the book profit function}\label{section_book_profits}

In this section, we state some properties of $F(t,x)$. These are needed in the next section for showing convergence of $\Pi^n$.

\begin{proof}[Proof of Proposition~\ref{properties_F}]
(i):  Let $y\le x\leq\varphi_t$. By (\ref{19.9.2013.1}), we have $M_{t,x}\not=\emptyset$. By the left-continuity of $\varphi$, $\sup M_{t,x}$ is attained, i.e., 
$x-\varphi_{t}+\varphi_{\tau_{t,x}}\leq 0$ or $x-\varphi_{t}+\varphi_{\tau_{t,x}+}\leq 0$. We conclude that 
$y-\varphi_{t}+\varphi_{\tau_{t,x}}\leq 0$ or $y-\varphi_{t}+\varphi_{\tau_{t,x}+}\leq 0$. Thus $\tau_{t,x}\leq \tau_{t,y}$.\vspace{0.1cm}\\
(ii): Follows immediately from (\ref{25.7.2013.3}).\vspace{0.1cm}\\
(iii): Due to $\tau_{t,y}\ge \tau_{t,x}$ for $y\le x\le\varphi_t$, one has that 
$F(t,x)-F(t,y)=\inf_{\tau_{t,y}\leq u\leq t}S_u-\inf_{\tau_{t,x}\leq u\leq t}S_u\geq 0$.\vspace{0.1cm}\\
(iv): By (ii), it is enough to show left-continuity at $x\in(0,\varphi_t]$. 
One has $x-\varphi_t+\varphi_u>0$ for all $u\in(\tau_{t,x},t]$ and $x-\varphi_t+\varphi_{u+}>0$ for all $u\in(\tau_{t,x},t)$. 
Because the infimum of a c\`agl\`ad process on a compact interval is attained in a right or a left limit, one has that 
\beao
\inf\left\{ x-\varphi_t+\varphi_u\ |\ u\in[\tau_{t,x}+\varepsilon,t]\right\}>0,\quad \forall\varepsilon>0.
\eeao
Therefore, there exists $\delta_0>0$ s.t. for all $\delta\in(0,\delta_0]$  
\beao
x-\delta-\varphi_t+\varphi_u>0\quad \forall u\in[\tau_{t,x}+\varepsilon,t]\quad\mbox{and}\quad x-\delta-\varphi_t+\varphi_{u+}>0\quad \forall u\in[\tau_{t,x}+\varepsilon,t).
\eeao
Thus, either $M_{t,x-\delta}=\emptyset$ or $0\le \sup M_{t,x-\delta} \le \tau_{t,x} + \eps$. If the first holds for some $\delta\in(0,\delta_0]$, it also holds for all smaller positive 
numbers and zero. In this case, left-continuity is obvious 
because $\tau_{t,y}=\tau_{t,x}=t$ for all $y$ in a left neighborhood of $x$. 
In the second case, one has $\tau_{t,x-\delta}-\tau_{t,x}\leq\varepsilon$ and,
by (i), $\tau_{t,x-\delta}\in[\tau_{t,x},\tau_{t,x}+\varepsilon]$ for all $\delta\in(0,\delta_0]$. 
By right-continuity of $S$ we are done.\vspace{0.1cm}\\
(v): Let $\varphi$ be an elementary strategy with representation as in Definition~\ref{Pi_elementar}.
Let $t\in[\kappa_{i-1},\kappa_i)$ and $s_1,s_2\in (t,\kappa_i]$, i.e., $\varphi_{s_1}=\varphi_{s_2}$. For $x=0$, one has $F(s_1,0)=F(s_2,0)=0$. For
$x\in(0,\varphi_{s_1}]$, one has $M_{s_1,x},M_{s_2,x}\subset [0,\kappa_{i-1}]$ which leads, again by 
$\varphi_{s_1}=\varphi_{s_2}$, to $M_{s_1,x}=M_{s_2,x}$.
By $x\le \varphi_{s_1}$ and (\ref{19.9.2013.1}), one has $M_{s_1,x}\not=\emptyset$ and arrives at $\tau_{s_1,x}=\tau_{s_2,x}\le \kappa_{i-1}$ and thus $F(s_1,x)=F(s_2,x)$.
For $x> \varphi_{s_1}=\varphi_{s_2}$ one has that $M_{s_1,x}=M_{s_2,x}=\emptyset$ and thus $F(s_1,x)=F(s_2,x)=0$. 
Consequently, the limit $\lim_{s\downarrow t}\tau_{s,x}=:\tau_{t+,x} $ exists for all $x\in\mathbb{R}_+$. 
\end{proof}

In the next lemma, we examine the behavior of the book profit function for two strategies whose paths are close together.
\begin{thmL}\label{lemma44}
Let $\varphi,\wt{\varphi}\in\mathbb{L}$ with $\varphi_0=\wt{\varphi}_0=0$ and $\varphi,\wt{\varphi}\ge 0$. $\wt{\tau}_{t,x}$, $\wt{F}$, and $\wt{M}_{t,x}$ 
denote the quantities from Definition~\ref{bookprofits} for $\wt{\varphi}$ instead of $\varphi$. Fix some $\omega\in\Omega$ and $t\in[0,T]$. If 
\beam\label{30.6.2013.1}
\sup_{0\leq u\leq t}|\varphi_u(\omega)-\wt{\varphi}_u(\omega)|\le \varepsilon,
\eeam
then 
\beam\label{Fn}
F_\omega(t,x)\leq \wt{F}_\omega(t,x+2\varepsilon)\quad \mbox{for all}\ x\leq\wt{\varphi}_t(\omega)-2\varepsilon\quad\mbox{and}
\eeam
\beam\label{30.6.2013.2}
\left|\int_0^{\varphi_t(\omega)}F_\omega(t,x)\,dx- \int_0^{\wt{\varphi}_t(\omega)} \wt{F}_\omega(t,x)\,dx\right|
\leq 3\varepsilon\left(\sup_{0\leq u\leq t} S_u(\omega)-\inf_{0\leq u\leq t} S_u(\omega)\right).
\eeam
\end{thmL}
\begin{proof}
We fix some $\omega\in\Omega$ satisfying (\ref{30.6.2013.1}) and omit it in the rest of the proof. 
Let $x\le \wt{\varphi}_t - 2\eps$. By (\ref{30.6.2013.1}), one has $\wt{M}_{t,x+2\eps}\subset M_{t,x}$.
This gives $\sup \wt{M}_{t,x+2\eps}\le \sup M_{t,x}$. Furthermore, by (\ref{19.9.2013.1}), one has $\wt{M}_{t,x+2\eps}\not=\emptyset$ and  thus 
$\wt{\tau}_{t,x+2\eps}= \sup \wt{M}_{t,x+2\eps} \le \sup M_{t,x} \le \tau_{t,x}$, which implies
\beao
F(t,x)-\wt{F}(t,x+2\varepsilon) = \inf_{\wt{\tau}_{t,x+2\varepsilon}\leq u\leq t} S_u -\inf_{\tau_{t,x}\leq u\leq t} S_u\leq 0.
\eeao
As obviously $F(t,x)=S_t-\inf_{\tau_{t,x}\leq u\leq t} S_u\leq \sup_{0\leq u\leq t} S_u-\inf_{0\leq u\leq t} S_u$ for all $x\in\mathbb{R}_+$, (\ref{Fn}) implies
\begin{align*}
 \int_0^{\varphi_t}F(t,x)\,dx\leq &\int_0^{\left(\wt{\varphi}_t-2\varepsilon\right)\vee 0}\wt{F}(t,x+2\varepsilon)\,dx
 +\left(\varphi_t-\wt{\varphi}_t+2\varepsilon\right)\left(\sup_{0\leq u\leq t} S_u-\inf_{0\leq u\leq t} S_u\right)\\
			       \leq &\int_0^{\wt{\varphi}_t}\wt{F}(t,x)\,dx+3\varepsilon\left(\sup_{0\leq u\leq t} S_u-\inf_{0\leq u\leq t} S_u\right).
\end{align*}
By symmetry, we obtain (\ref{30.6.2013.2}).
\end{proof}

In the next section, we prove that $\Pi$ is an optional process. For this purpose, some measurability of $F$ has to be checked.

\begin{thmP}\label{prop44}
$F$ is $\mathcal{O}\otimes \mathcal{B}(\mathbb{R}_+)-\mathcal{B}(\mathbb{R}_+)-$measurable.
\end{thmP}
\begin{proof}
Because $x\mapsto F_\omega(t,x)$ is left-continuous and on $[0,\varphi_t]$ also nondecreasing, one gets
\begin{align*}
F_\omega(t,x) = 1_{(x\le \varphi_t(\omega))} \sup_{q\in\mathbb{Q}_+}\left\{F_\omega(t,q)-1_{(x< q)}\infty\right\}.
\end{align*}
As $\{(\omega,t,x)\ |\ x\le \varphi_t(\omega)\}\in\mathcal{P}\otimes\mathcal{B}(\mathbb{R}_+)$, it remains to show that $(\omega,t)\mapsto F_\omega(t,q)$ is 
$\mathcal{O}-\mathcal{B}(\mathbb{R}_+)-$measurable for every fixed $q$.\\ 

{\em Step 1:} Let us show that $(\omega,t)\mapsto\tau_{t,q}(\omega)$ is $\mathcal{P}-\mathcal{B}(\mathbb{R}_+)-$measurable. Define the (random) sets 
\beao
M_{t,q}^n:=\{u\in[0,t]\ |\ q-\varphi_t+\varphi_u\le 1/n,~{u\in\mathbb{Q}}\},\quad n\in\mathbb{N}.
\eeao
By
\beao
\sup M_{\cdot,q}^n & = & \sup_{u\in\mathbb{Q}_+} u1_{\{(\omega,t)\ |\ q-\varphi_t(\omega) + \varphi_u(\omega)\le 1/n\ 
\mbox{\small and}\ u<t\}} 
\eeao
and the predictability of $\varphi$, the mapping $\sup M^n_{\cdot,q}: \Omega\times[0,T]\to\mathbb{R}_+,\ (\omega,t)\mapsto \sup M^n_{t,q}(\omega)$ is written as a pointwise
supremum over countably many predictable functions, and thus it is also predictable. Now, it is shown that 
\beam\label{25.7.2013.1}
\sup M_{t,q}^n 1_{(q\le \varphi_t)}\rightarrow\tau_{t,q}1_{(q\le \varphi_t)}\quad\mbox{pointwise for\ }n\to\infty.
\eeam
Let $n\in\mathbb{N}$,\ $u\in M_{t,q}$. There exists $v\in\mathbb{Q}$ arbitrary close to $u$ with $v\in M^n_{t,q}$ and thus
\beam\label{25.7.2013.2}
\sup M_{t,q} \le \sup M^n_{t,q},\quad\forall n\in\mathbb{N}.
\eeam

Assume that $q\le\varphi_t$, i.e., $\tau_{t,q}=\sup M_{t,q}$ by (\ref{25.7.2013.3}). First note that $q-\varphi_t+\varphi_u>0$ for all $u\in(\tau_{t,q},t]$ and 
$q-\varphi_t+\varphi_{u+}$ for all $u\in(\tau_{t,q},t)$. 
As the infimum of a c\`adl\`ag process is attained in the right or the left limit on a compact interval, one has that 
\beao
\inf\{q-\varphi_t+\varphi_u\ |\ u\in[\tau_{t,q}+\varepsilon,t]\}>0,\quad \forall \varepsilon>0.
\eeao
Therefore, there exists $N\in\mathbb{N}$ s.t. 
\beao
q-\frac1n-\varphi_t+\varphi_u>0\quad\forall u\in[\tau_{t,q}+\varepsilon,t],\ n\geq N.
\eeao
This implies $\sup M^n_{t,q} \le \tau_{t,q} + \eps = \sup M_{t,q} + \eps$ for all $n\ge N$. Together with (\ref{25.7.2013.2}) one obtains (\ref{25.7.2013.1}).
(\ref{25.7.2013.1}), the predictability of $\sup M^n_{\cdot,q}$, and (\ref{25.7.2013.3}) imply the predictability of $(\omega,t)\mapsto \tau_{t,q}(\omega)$.

{\em Step 2:} One has
\begin{align*}
F(t,q)	=& S_t -\inf_{\tau_{t,q}\leq u\leq t} S_u = 0\vee\sup_{y\in\mathbb{Q}}(S_t-S_y)1_{(\tau_{t,q}<y<t)}
\end{align*}
and by Step~1 $\{(\omega,t)\ |\ \tau_{t,q}(\omega)<y\}\in\mathcal{P}$. Because $S$ is optional, $F(\cdot,q)$ is also optional, which completes the proof.
\end{proof}

\section{Proof of Theorem~\ref{main_theorem}}\label{section_proofmain}

\begin{thmP}\label{prop1}
For any elementary strategy~$\varphi$, it holds that
\beam\label{27.7.2013.1}
\alpha \int_0^t\varphi_u dS_u+\alpha \int_0^t \varphi_u dD_u=\alpha\int_0^\infty F(t,x)dx+\Pi_t,\quad\forall t\in[0,T].
\eeam
\end{thmP}
 
This proposition is the key step to prove Theorem~\ref{main_theorem}. Namely, by the semimartingale property of $S$ and $D$ the integrals converge if $\varphi^n\to\varphi$, and with 
Lemma~\ref{lemma44} it can be shown that also the corresponding book profits $\int_0^\infty F(t,x)dx$ converge. 
For the latter one needs that $\varphi^n$ converges uniformly in probability and not only pointwise.  
To prove the proposition one needs the following lemma.
\begin{thmL}\label{1.9.2013.1}
Let $\varphi$ be an elementary strategy, s.t. $\varphi=\sum_{i=1}^{k}H_{i-1} 1_{\rrbracket \kappa_{i-1},\kappa_i\rrbracket},$ where 
$0=\kappa_0\le\kappa_1\le \ldots\le \kappa_k=T$ are stopping times and $H_{i-1}$ is $\mathcal{F}_{\kappa_{i-1}}-$measurable. For all $t\in(\kappa_{i-1},\kappa_i],\  x\in(0,\varphi_t]$, 
we have
\begin{align*}
 S_t-S_{\kappa_{i-1}}=\left(F(\kappa_{i-1}+,x)+\inf_{\kappa_{i-1}\leq u\leq t}(S_u-S_{\kappa_{i-1}})\right)\wedge 0+F(t,x)-F(\kappa_{i-1}+,x).
\end{align*}
\end{thmL}
\begin{proof}
Let $t_1,t_2\in (\kappa_{i-1},\kappa_i]$, i.e., $\varphi_{t_1}=\varphi_{t_2}$. As $x>0$, one has $M_{t_1,x},M_{t_2,x}\subset [0,\kappa_{i-1}]$, which leads, again by 
$\varphi_{t_1}=\varphi_{t_2}$, to $M_{t_1,x}=M_{t_2,x}$.
By $x\le \varphi_{t_1}$, we have $0\in M_{t_1,x}\not=\emptyset$ and arrive at 
\beam\label{3.7.2013.1}
\tau_{t_1,x}=\tau_{t_2,x}\le \kappa_{i-1}.
\eeam
By (\ref{3.7.2013.1}), the limit $\lim_{s\downarrow \kappa_{i-1}}\tau_{s,x}=:\tau_{\kappa_{i-1}+,x} $ exists and coincides with $\tau_{t,x}$,\ $t\in(\kappa_{i-1},\kappa_i]$. 
This leads to
\beam\label{3.7.2013.2}
 \left(-\inf_{\tau_{\kappa_{i-1}+,x}\leq u\leq \kappa_{i-1}}S_u+\inf_{\kappa_{i-1}\leq u\leq t}S_u\right)\wedge 0 
&=& \left(-\inf_{\tau_{t,x}\leq u\leq \kappa_{i-1}}S_u+\inf_{\kappa_{i-1}\leq u\leq t}S_u\right)\wedge 0\nonumber\\
&=& -\inf_{\tau_{t,x}\leq u\leq \kappa_{i-1}}S_u+\inf_{\tau_{t,x}\leq u\leq t}S_u\nonumber\\
&=& -\inf_{\tau_{\kappa_{i-1}+,x}\leq u\leq \kappa_{i-1}}S_u+\inf_{\tau_{t,x}\leq u\leq t}S_u,
\eeam
where for the second equality we use that, by (\ref{3.7.2013.1}), $[\tau_{t,x},t] = [\tau_{t,x},\kappa_{i-1}]\cup [\kappa_{i-1},t]$, and we distinguish the cases 
$\inf_{\tau_{t,x}\leq u\leq \kappa_{i-1}}S_u \ge \inf_{\kappa_{i-1}\leq u\leq t}S_u$ and $\inf_{\tau_{t,x}\leq u\leq \kappa_{i-1}}S_u < \inf_{\kappa_{i-1}\leq u\leq t}S_u$.
Using (\ref{3.7.2013.2}), the right-continuity of $S$, and the definition of $F$, it can immediately be seen that the LHS of (\ref{3.7.2013.2}) equals
$$\left(F(\kappa_{i-1}+,x)+\inf_{\kappa_{i-1}\leq u\leq t}(S_u-S_{\kappa_{i-1}})\right)\wedge 0,$$
and the RHS of (\ref{3.7.2013.2}) equals
$$F(\kappa_{i-1},x)-F(t,x)+S_t-S_{\kappa_{i-1}}.$$
So we are done.
\end{proof}

\begin{proof}[Proof of Proposition~\ref{prop1}]
Let $\varphi$ be as in Lemma~\ref{1.9.2013.1}. First, we consider increments of (\ref{27.7.2013.1}) on $(\kappa_{i-1},\kappa_i]$,\ $i\in\{1,\ldots,k\}$. 
Let $t_1,t_2\in(\kappa_{i-1},\kappa_i]$. Because $\varphi_{t_1}=\varphi_{t_2}$ on $(\kappa_{i-1},\kappa_i]$, one has by definition of $\Pi$

\begin{align*}
\Pi_{t_2}-\Pi_{t_1}=&\alpha\int_0^{\varphi_{t_1}}\left(F(\kappa_{i-1}+,x)+\inf_{\kappa_{i-1}\leq u\leq t_2\wedge \kappa_i}(S_u-S_{\kappa_{i-1}})\right)\wedge0dx\\
		    &-\alpha\int_0^{\varphi_{t_1}}\left(F(\kappa_{i-1}+,x)+\inf_{\kappa_{i-1}\leq u\leq t_1\wedge \kappa_i}(S_u-S_{\kappa_{i-1}})\right)\wedge0dx
		    +\alpha\varphi_{t_1}(D_{t_2}-D_{t_1}).
\end{align*}
By Lemma~\ref{1.9.2013.1}, one arrives at
\begin{align*}
\Pi_{t_2}-\Pi_{t_1}=&\alpha\int_0^{\varphi_{t_1}}\left(S_{t_2}-S_{\kappa_{i-1}} -F(t_2,x)+F(\kappa_{i-1}+,x)\right)dx\\
		    & -\alpha\int_0^{\varphi_{t_1}}\left(S_{t_1}-S_{\kappa_{i-1}} -F(t_1,x)+F(\kappa_{i-1}+,x)\right)dx
		    +\alpha \varphi_{t_1}(D_{t_2}-D_{t_1})\\
		   =&\alpha\varphi_{t_1}\left(S_{t_2}+D_{t_2}-S_{t_1}-D_{t_1}\right) -\alpha\int_0^\infty (F(t_2,x)-F(t_1,x))dx\\
		   =&\alpha\int_0^{t_2}\varphi_sd(S+D)_s-\alpha\int_0^{t_1}\varphi_sd(S+D)_s -\alpha\int_0^\infty F(t_2,x)dx + \alpha\int_0^\infty F(t_1,x)dx,
\end{align*}
where in the last equality we use that $\varphi_s=\varphi_{t_1}$ for all $s\in(t_1,t_2]$.
This means that (\ref{27.7.2013.1}) holds true for all increments on $(\kappa_{i-i},\kappa_i]$. As it obviously holds for $t=0$, it remains to show that the right jumps of the
processes $t\mapsto \int_0^\infty F(t,x)dx$ and $\Pi$ at $\kappa_{i-1}$ sum up to $0$ as the LHS of (\ref{27.7.2013.1}) is right-continuous. By similar arguments as in the proof of
Proposition~\ref{properties_F}(v), one obtains 
\beam\label{29.7.2013.2}
\tau_{\kappa_{i-1}+,x}=\tau_{\kappa_{i-1},x-(\varphi_{\kappa_{i-1}+}-\varphi_{\kappa_{i-1}})}\quad \forall x\in\mathbb{R}_+\quad\mbox{with 
the convention $\tau_{\kappa_{i-1},y}:=t\ \forall y<0$}.
\eeam
With the convention $F(\kappa_{i-1},y)=0$ for $y<0$, one obtains
\beam\label{29.7.2013.1}
\lim_{t\downarrow \kappa_{i-1}}\int_0^{\varphi_t}F(t,x)dx & = & \int_0^{\varphi_{\kappa_{i-1}+}}\left(S_{\kappa_{i-1}} 
- \inf_{\tau_{\kappa_{i-1}+,x}\le u\le \kappa_{i-1}}S_u\right)dx\nonumber\\
& \stackrel{(\ref{29.7.2013.2})}{=} & \int_0^{\varphi_{\kappa_{i-1}+}}F(\kappa_{i-1},x-\Delta^+\varphi_{\kappa_{i-1}})dx\nonumber\\
& = & \int_{-\Delta^+\varphi_{\kappa_{i-1}}}^{\varphi_{\kappa_{i-1}+}-\Delta^+\varphi_{\kappa_{i-1}}}F(\kappa_{i-1},x)dx\nonumber\\
& = & \int_0^{\varphi_{\kappa_{i-1}}}F(\kappa_{i-1},x)dx-\int_0^{-\Delta^+\varphi_{\kappa_{i-1}}}F(\kappa_{i-1},x)dx\nonumber\\
& = & \int_0^{\varphi_{\kappa_{i-1}}}F(\kappa_{i-1},x)dx-\int_0^{\left(\Delta^+\varphi_{\kappa_{i-1}}\right)^-}F(\kappa_{i-1},x)dx,
\eeam
where the first equality follows from the definition of $F$ using that $S$ is right-continuous and $\tau_{t,x}=\tau_{\kappa_{i-1}+,x}$ for all $t\in(\kappa_{i-1},\kappa_i]$ and $x>0$.
(\ref{29.7.2013.1}) means that 
\begin{align}
- \Delta^+\Pi_{\kappa_{i-1}}=\Delta^+\left(\int_0^{\varphi_{\kappa_{i-1}}}F(\kappa_{i-1},x)dx\right) 
=  -\int_0^{\left(\Delta^+\varphi_{\kappa_{i-1}}\right)^-}F(\kappa_{i-1},x)dx,\label{Pi_right_jump}
\end{align}
 and we are done.
\end{proof}

\begin{proof}[Proof of Theorem \ref{main_theorem}]
{\em Step 1:} Let $(\varphi^n)_{n\in\mathbb{N}}$ be a sequence of nonnegative elementary strategies with $\varphi^n_0=0$ and
$\varphi^n\stackrel{\textrm{up}}{\rightarrow}\varphi$.
>From Proposition~\ref{prop44} one knows that $(\omega,t,x)\mapsto F^n_{\omega}(t,x)$ is $\mathcal{O}\otimes\mathcal{B}(\mathbb{R}_+)-\mathcal{B}(\mathbb{R}_+)$-measurable. 
So, $(\omega,t)\mapsto\int_0^\infty F^n_{\omega}(t,x)dx$ is $\mathcal{O}-\mathcal{B}(\mathbb{R}_+)-$measurable. 
Together with Proposition~\ref{prop1} and the fact that $\varphi^n\mal S$ and $\varphi^n\mal D$ are optional, this implies that $\Pi^n$ is also optional. 

In the next step, it is shown that $(\Pi^n)_{n\in\mathbb{N}}$ is an up-Cauchy sequence. Again by Proposition~\ref{prop1}, it is enough to show that 
$(\varphi^n\mal S)_{n\in\mathbb{N}}$, $(\varphi^n\mal D)_{n\in\mathbb{N}}$, and $\left(\int_0^\infty F^n(\cdot,x)dx\right)_{n\in\mathbb{N}}$ are up-Cauchy sequences. 
Because $\varphi^n\stackrel{\textrm{up}}{\rightarrow}\varphi$ and $S$, $D$ are semimartingales, it is known, e.g., from Theorem~II.11 in \cite{protter1}, 
that $(\varphi^n \mal S)_{n\in\mathbb{N}}$, $(\varphi^n\mal D)_{n\in\mathbb{N}}$ are up-Cauchy sequences. So, it remains to consider $\int_0^\infty F^n(t,x)dx$. \vspace{0.3cm}\\
\indent Let $\varepsilon>0$. As $S$ possesses c\`adl\`ag paths, there exists $K\in\mathbb{R}_+$ s.t. 
$$P\left(\sup_{0\le t\le T}S_t - \inf_{0\le t\le T} S_t\geq K\right)\leq \frac\varepsilon 2.$$
As $\varphi^n\stackrel{\mbox{\small up}}{\rightarrow}\varphi$, there exists $N_{\varepsilon}\in\mathbb{N}$ s.t.
$$P\left(\sup_{0<t\leq T}|\varphi_t^n-\varphi_t^m|> \frac{\varepsilon}{3K}\right)\leq \frac{\varepsilon}2,\quad \forall n,m\geq N_{\varepsilon}.$$
By Lemma \ref{lemma44}, we have 
\beao
\left\{\sup_{0\le t\leq T}\left|\int_0^\infty F^n(t,x)-F^m(t,x)dx\right|>\frac{\varepsilon}{K}(\sup_{0\le t\le T}S_t-\inf_{0\le t\le T} S_t)\right\}
\subset
\left\{\sup_{0<t\leq T}|\varphi_t^n-\varphi_t^m| > {\frac{\varepsilon}{3K}}\right\},
\eeao
and one gets
\begin{align*}
 &P\left(\sup_{0\le t\leq T}\left|\int_0^\infty F^m(t,x)-F^n(t,x)\right|>\varepsilon\right)\\
\leq & P\left((\sup_{0\le t\le T} S_t-\inf_{0\le t\le T}S_t)\frac\varepsilon K>\varepsilon\right)+ P\left(\sup_{0<t\leq T} |\varphi^n_t-\varphi^m_t|\geq \frac\varepsilon{3K}\right)
\leq \frac{\varepsilon}2 + \frac{\varepsilon}2 = \varepsilon\quad\forall n,m\ge N_\eps.
\end{align*}
So, $(\Pi^n)_{n\in\mathbb{N}}$ is an up-Cauchy sequence. Because the space of l\`agl\`ad functions (also called ``regulated functions'') 
mapping from $[0,T]$ to $\mathbb{R}$ is complete w.r.t. the supremum norm, 
there exists an optional l\`agl\`ad process~$\Pi$ s.t.
$\Pi^n\stackrel{\mbox{\small up}}{\to}\Pi$ (optionality follows from pointwise convergence up to evanescence of a suitable subsequence and the usual conditions).\vspace{0.3cm}\\

{\em Step 2:} Let us now show (\ref{jumpPi-}). Let $t\in(0,T]$, $x_0\in(0,\varphi_t)$ and assume that 
\beao
x\mapsto \wt{F}(t,x):= S_{t-} - \inf_{\tau_{t,x}\le u<t} S_u
\eeao
is continuous at $x_0$. $\wt{F}(t,\cdot)$ is the time-$t$ book profit function under the modified stock price process $\wt{S}_u:=1_{(u<t)}S_u + 1_{(u\ge t)}S_{t-}$ (this 
modification removes the impact of $\Delta S_t$ on the book profits).

Let $\varepsilon\in(0,\varphi_t-x_0)$. By the left-continuity of $\varphi$ and by $\tau_{t,x_0+\varepsilon}\le \tau_{t,x_0}<t$, one has for $s$ smaller but close enough to $t$ that
\beam\label{22.8.2014.1}
|\vp_s - \vp_t|\le \eps\quad\mbox{and}\quad s>\tau_{t,x_0+\eps}.
\eeam
For $s$ satisfying (\ref{22.8.2014.1}), one has
that $M_{s,x_0}\not=\emptyset$, $M_{t,x_0+\eps}\cap[0,s]\not=\emptyset$, and the two implications
\beao
u\in M_{s,x_0} \Rightarrow  u\in M_{t,x_0-\eps},\qquad u\in M_{t,x_0+\eps}\cap[0,s] \Rightarrow  u\in M_{s,x_0}
\eeao
hold; see (\ref{13.8.2014.1}) for the definition of $M$. This implies
\beao
\tau_{t,x_0-\varepsilon} \ge \tau_{s,x_0} \ge \tau_{t,x_0+\varepsilon}.
\eeao
It follows that 
\beao
\inf_{\tau_{t,x_0+\varepsilon}\le u<t} S_u \le \inf_{\tau_{s,x_0}\le u<t} S_u \le \inf_{\tau_{t,x_0-\varepsilon}\le u<t} S_u.
\eeao
By the continuity of $\wt{F}(t,\cdot)$ in $x_0$, the left and the right bound are close together for $\varepsilon$ small. 
We conclude that $\lim_{s<t,s\to t}F(s,x_0)=:F(t-,x_0)$ exists and 
\beam\label{27.9.2013.1}
F(t-,x_0)=S_{t-}-\inf_{\tau_{t,x_0}\leq u<t} S_u
\eeam
(For elementary strategies, one has that $\tau_{s,x}=\tau_{t,x}$ for $s$ smaller but close to $t$, and therefore the limit $F(t-,x)$ exists for all $x\in\bbr_+$). 
By (\ref{27.9.2013.1}) and a distinction of the cases $S_t<\inf_{\tau_{t,x_0}\leq u<t} S_u$ and $S_t\ge \inf_{\tau_{t,x_0}\leq u<t} S_u$, one obtains 
$F(t,x_0) = 0 \vee (F(t-,x_0)+\Delta S_t)$ and thus
\beao
\Delta F(t,x_0) = (-F(t-,x_0))\vee \Delta S_t = \Delta S_t  + (-F(t-,x_0)-\Delta S_t)\vee 0.
\eeao 
By monotonicity, the mapping~$x\mapsto \wt{F}(t,x)$  has at most countably
many discontinuities, so that $\lim_{s<t,s\to t} \int_0^\infty F(s,x)\,dx$ exists and  
 \beam\label{11.9.2013.3}
\Delta \int_0^\infty F(t,x)\,dx = \Delta \int_0^{\varphi_t} F(t,x)\,dx = \varphi_t \Delta S_t + \int_0^{\varphi_t}(-\limsup_{s<t,s\to t}F(s,x)-\Delta S_t)\vee 0\,dx 
\eeam
(interchanging integral and limit is possible as $F$ and $S$ are bounded for $\omega$ fixed). 
By construction of $\Pi$, Proposition~\ref{prop1} holds for all $\varphi\in\bbl$. Together with (\ref{11.9.2013.3}) and $\Delta (\varphi\mal (S+D))=\varphi\Delta(S+D)$, 
this implies (\ref{jumpPi-}).
 
{\em Step 3:} It remains to prove (\ref{jumpPi+}). For the approximating elementary trading strategies~$\varphi^n$, it follows immediately from Definition~\ref{Pi_elementar}. 
As $(\Delta^+ \varphi^n)^-$ converges to $(\Delta^+ \varphi)^-$ uniformly in probability,  
\begin{align}
\int_0^{(\Delta^+\varphi_t^{n})^-}F^{n}(t,x)dx\stackrel{\mbox{up}}{\longrightarrow}\int_0^{(\Delta^+\varphi_t)^-}F(t,x)dx\label{gleichung1}
\end{align}
follows by the same arguments as in the proof of Lemma~\ref{lemma44}. Putting everything together the assertion follows.
\end{proof}

\section{Self-financing condition}\label{section_self}
To prepare Section \ref{section_comparison}, we introduce the self-financing condition of the model which is a natural generalization of the standard continuous time 
self-financing condition without taxes.

Besides the risky stock with price process~$S$  and dividend process~$D$, the market consists of a so-called bank account. 
Formally, the bank account can be seen as a security with price
process~$1$ and {\em dividend process}
\begin{align}
B_t=\int_0^t r_s\,ds,\quad t\in[0,T],\label{bond_dividend}
\end{align}
where the locally riskless interest rate $r$ is a predictable, {\em nonnegative}, and integrable process. This simplifies the analysis as increments of $B$ are taxed immediately, 
and one needs not consider unrealized book profits of the bank account (as for the risky stock).
\begin{thmD}[Wealth process and self-financing condition]
Let $X$ be an optional process  modeling the number of monetary units in the bank account, and
$\varphi\in\mathbb{L}$ models the number of stocks the investor holds in her portfolio. 
The {\em wealth process} $V$ of the strategy~$(X,\varphi)$ is defined as
\beam\label{4.8.2013.1}
V=V(X,\varphi):=X + \varphi S.
\eeam
A strategy $(X,\varphi)$ is called self-financing with initial wealth $v_0$ iff
\beam\label{4.8.2013.2}
V=v_0+(1-\alpha)X\mal B + \varphi\mal D+\varphi\mal S-\Pi
\eeam
with $\Pi$ from Definition~\ref{6.8.2013.1}. 
\end{thmD}
\begin{Bemerkung}
 As $B$ is continuous, it is sufficient to assume that $X$ is optional instead of predictable. Thus, the after-tax dividend $(1-\alpha)\varphi_t\Delta D_t$ of the stock can be included 
in the number of monetary units~$X_t$. 
Note that an immediate reinvestment of the payoff in the stock would only affect $\varphi_{t+}$, but not $\varphi_t$.
\end{Bemerkung}
\begin{Bemerkung}\label{10.9.2013.1}
For any $\varphi\in\mathbb{L}$, $v_0\in\bbr$, there exists a unique optional process~$X$ s.t. $(X,\varphi)$ is self-financing. 
Indeed, plugging (\ref{4.8.2013.1}) into (\ref{4.8.2013.2}) yields
\begin{align}\label{SDE_X}
 X=v_0+(1-\alpha)X\mal B + \varphi\mal D+\varphi\mal S-\Pi-\varphi S.
\end{align}
Now, an optional process $X$ solves (\ref{SDE_X}) iff $X$ is l\`agl\`ad, the c\`adl\`ag process~$X_+$ solves the SDE
$$Z= v_0 + (1-\alpha)Z_-\mal B + \varphi\mal D+\varphi\mal S-\Pi_+-\varphi_+S$$
(which has a unique solution~$Z$, cf., e.g., Theorem V.7 in \cite{protter1}), and $X=Z-\Delta^+ X=Z+\Delta^+ \Pi+S\Delta^+\varphi$. 
\end{Bemerkung}
(\ref{4.8.2013.2}) means that increments of the wealth process solely result from trading gains and tax payments. An alternative condition is to assume 
that portfolio regroupings do not 
involve costs. The latter condition may be more intuitive, but it has the drawback that it can only be stated for strategies that can be used as integrators
(thus, trading strategies that are no semimartingales would be excluded although they could economically make sense).
Let $\varphi$ and $\Pi$ be as in Definition~\ref{Pi_elementar}. The alternative self-financing condition reads
\begin{align}
X_t=v_0-\sum_{i=1}^k 1_{(\kappa_{i-1}< t)} S_{\kappa_{i-1}}(\varphi_{\kappa_{i-1}+}-\varphi_{\kappa_{i-1}})+\int_0^t(1-\alpha)X_s r_s ds -\Pi_t+\varphi\mal D_t.\label{f3}
\end{align}

It is an easy exercise to prove equivalence of (\ref{f3}) and (\ref{4.8.2013.2}) for elementary strategies.

\section{Comparison of different dividend policies}\label{section_comparison}

In this section, we investigate the effect of different dividend policies on the investor's after-tax wealth. 
In particular, we show that under the mild condition that the dividend policy has no effect on the {\em stochastic return process}, the effect of dividends is always negative.
This assumption is formalized by the following definition.

\begin{Definition}\label{25.9.2013.1}
Let $R$ be a semimartingale with $\Delta R\ge -1$ and $s_0\in\bbr_+$. Then, for any nondecreasing c\`adl\`ag process~$D$, define $S^D$ as the unique solution of 
\beam\label{9.9.2013.1}
S^D = s_0 + S^D_-\mal R - D.
\eeam
We call $D$ admissible iff $S^D\ge 0$, i.e., we only consider dividend payoffs that do not exceed the stock price. 
$R$ is the {\em return process} modeling the stochastic 
profit per invested capital. 
\end{Definition}
Observe that for any admissible $D$ the stock price $S^D$ stays at zero once the process or its left limit hit it. Note that by $\Delta R\ge -1$, $D=0$, which corresponds to the model 
without dividends, is admissible. 
Alternatively, one can start with an arbitrary nondecreasing process~$\wt{D}$ with
\beam\label{1.8.2013.1}
\Delta \wt{D} \le 1+ \Delta R
\eeam
modeling accumulated dividends as multiples of the current stock price and consider the SDE
\beam\label{4.8.2013.3}
S = s_0 + S_-\mal (R - \wt{D}).
\eeam
Then, $S^D=S$ for $D := S_-\mal \wt{D}$, and, by (\ref{1.8.2013.1}), the stock price is nonnegative. But, as for an arbitrary admissible dividend process~$D$ the integral
$\frac1{S^D_-}1_{\{S^D_->0\}}\mal D$ may explode, Definition~\ref{25.9.2013.1} is slightly more general.\\ 

\begin{Bemerkung}\label{21.9.2013.1}
(\ref{9.9.2013.1}) says that one has the same $R$ for all processes~$D$, i.e., there holds a scaling invariance of the stochastic investment opportunities. 
The negative effect of dividends on the after-tax wealth is essentially based on this property. It is, e.g., not satisfied in the Bachelier model with dividends.

Note that we do not assume that dividend payoffs are accompanied by downward jumps of the same size of the stock price. Such a behavior can be explained by no-arbitrage 
arguments if dividends are predictable. However, the framework also allows for a spontaneous dividend payment~$\Delta D_t$, e.g., if $\Delta R_t$ is large.
\end{Bemerkung}


Recall that we consider a market model with two investment opportunities: a risky stock with price process~$S^D$ and dividend process~$D$
(interrelated by Condition (\ref{9.9.2013.1})) and a locally riskless bank account. The latter is an asset with price process~$1$ and the nondecreasing dividend process~$B$ 
from (\ref{bond_dividend}).  
We denote the model by $((S^D,D),(1,B))$. Now, we compare the situation of an arbitrary admissible dividend process~$D$ with the situation of no dividends. In  the latter model,
we use the subscript~$0$, i.e., $S^0$, $\Pi^0$, $V^0$, etc. The following theorem is the main result of this section.
\begin{thmS}\label{propV}
 Let $(X^D,\varphi^D)$ be a self-financing strategy with initial wealth~$v_0$ in the model with dividends $((S^D,D),(1,B))$, and let $V^D$ be the corresponding wealth process. 
 Then, there exists a self-financing strategy $(X^0,\varphi^0)$ with initial wealth~$v_0$ in the model without dividends $((S^0,0),(1,B))$, where $V^0$ is the 
 corresponding wealth process, s.t. $V^D\leq V^0$.
\end{thmS}

\begin{Lemma}\label{9.9.2013.5}
The process 
\beao
\frac{S^D}{S^0}1_{\{S^0>0\}}
\eeao
is nonincreasing.
\end{Lemma}
\begin{proof}
The case $s_0=0$ is obvious. Let $s_0>0$ and define $\tau:=\inf\{t\ge 0\ |\
\Delta R_t=-1\}$. By the formula of Yoeurp-Yor~\cite{yoeurp.yor.1977} (see also \cite{jaschke.2003}), one has
\beam\label{23.9.2013.1}
S^D = S^0 \left(1 - \frac1{S^0_-}\mal D + \sum_{0<s\le \cdot} \frac1{S^0_{s-}}\frac{\Delta D_s\Delta R_s}{1+\Delta R_s}\right)
\eeam
on the stochastic interval~$\auf 0,\tau\auf$. The second factor of the RHS of (\ref{23.9.2013.1}) is obviously a nonincreasing process. As
$S^D=S^0=0$ on $\auf \tau,\infty\auf$, we are done.
\end{proof}

The key step to prove Theorem~\ref{propV} is the following lemma.

\begin{Lemma}\label{9.9.2013.2}
Let $\varphi^D\in\mathbb{L}$ and $\varphi^0:=\varphi^D\frac{S^D_-}{S^0_-}1_{\{S^0_->0\}}$. Then, $\varphi^0\in\bbl$,
\begin{align}
\varphi^0\mal S^0 = \varphi^D\mal(S^D+D),\quad\mbox{and}\quad\Pi^0\leq\Pi^D.\label{ineq1}
\end{align}
\end{Lemma}
This means that for an arbitrary strategy in the model with dividends, there exists a strategy in the model without dividends leading to the same trading gains 
in the risky stock but not exceeding accumulated tax payments. The money invested in the stock is the same for both strategies.
If price processes do not vanish, one can recover $\varphi^D$ from $\varphi^0$ by investing the dividend payoffs in new stocks. 
This is  illustrated in Figure~\ref{ftx_dividends}. 

\begin{figure}[ht]
    \subfigure[Book profit functions \textcolor{red}{$x\mapsto F^D(t_1-,x)$} and \textcolor{blue}{$x\mapsto F^D(t_1,x)$} modeling book profits
    immediately before resp. after the predictable dividend 
    payoff~$\Delta D_{t_1}=1000$ associated with $\Delta S_{t_1}=-1000$. One has $F^D(t_1-,x)+\Delta S_{t_1}<0$ iff $x<55$. This means that 55 stocks are sold and immediately repurchased 
    (wash sale).]{\includegraphics[width=0.98\textwidth, height=120px]{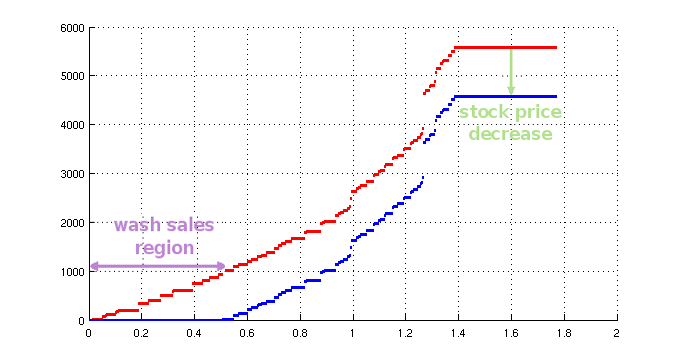}}\\
    \subfigure[Book profit function $x\mapsto F^D(t_1+,x)$  after portfolio regrouping. According to $\varphi^D$, the dividend payoff is invested in 20 new stocks which start 
    with zero book profits, and the function is shifted about 20 units to the right.]{\includegraphics[width=0.98\textwidth, height=120px]{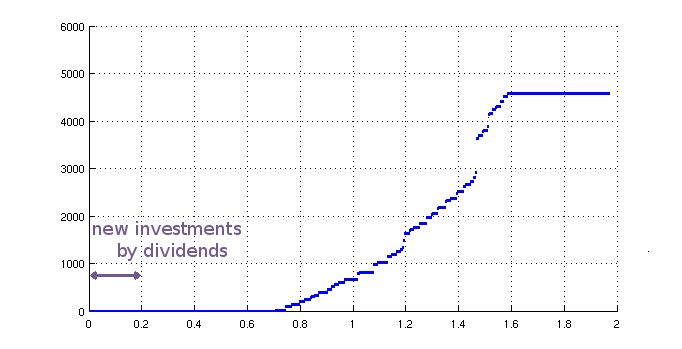}}
\caption{Reinvestment of dividends}\label{ftx_dividends}
\end{figure} 

\begin{proof}

{\em Step 1:} As $\varphi^0\le \varphi^D$, one obviously has $\varphi^0\in\bbl$. Because $S^D=S_-^D\mal R - D$ and $S^0=S^0_-\mal R$, one obtains
\beam\label{4}
\varphi^0\mal S^0 & = & \varphi^D\frac{S^D_-}{S_-^0}1_{\{S^0_->0\}}\mal (S^0_-\mal R) 
=\varphi^D S^D_- 1_{\{S^0_->0\}}\mal R = \varphi^D 1_{\{S^0_->0\}} \mal ( S^D_- \mal R)\nonumber\\
& = & \varphi^D\mal(S^D+D),
\eeam
where for the last equality we use that $\{S^0_-=0\}\subset \{S^D_-=0\}$ and the process  $S^D_- 1_{\{S^D_-=0\}} \mal R$ vanishes.

By construction of $\Pi$, Proposition \ref{prop1} holds for all strategies from $\mathbb{L}$, i.e.,
\begin{align*}
\alpha \int_0^\infty F^D(t,x)dx+\Pi_t^D=\alpha\varphi^D\mal(S^D+D)
\end{align*}
(and the same without dividends). Together with (\ref{4}), one obtains
\beam\label{9.9.2013.3}
\Pi_t^D-\Pi_t^0=\alpha\int_0^\infty F^0(t,x)dx-\alpha\int_0^\infty F^D(t,x)dx.
\eeam

{\em Step 2:} Let us show that for $\varphi^0_t>0$ (implying that $\varphi^D_t>0$ and $S^D_{t-}>0$)
\beam\label{9.9.2013.4}
\tau_{t,x}^D\geq \tau^0_{t,x\frac{\varphi_t^0}{\varphi_t^D}},\quad \forall x\in\bbr_+.
\eeam

First note that 
\begin{align}
M^0_{t,x\frac{\varphi_t^0}{\varphi_t^D}}=\emptyset\Leftrightarrow M^D_{t,x}=\emptyset\label{7}
\end{align}
(cf. Definition~\ref{bookprofits}). It is sufficient to consider $x$  s.t. both sets are not empty. One has
$$x-\varphi_t^D+\varphi_u^D>0\quad\forall u\in(\tau_{t,x}^D,t]\quad\mbox{and}\quad x-\varphi_t^D+\varphi_{u+}^D>0\quad\forall u\in(\tau_{t,x}^D,t).$$ 
We conclude
\begin{align*}
0<\frac{\varphi^0_t}{\varphi^D_t}\left(x-\varphi_t^D+\varphi_u^D\right)= \frac{\varphi^0_t}{\varphi^D_t}\left(x-\varphi_t^0\frac{S_{t-}^0}{S_{t-}^D}
+\varphi^0_u\frac{S_{u-}^0}{S_{u-}^D}\right)
\leq \frac{\varphi^0_t}{\varphi^D_t}x-\varphi_t^0 + \varphi^0_u\quad\forall u\in(\tau_{t,x}^D,t],
\end{align*}
where for the last inequality we use that $\frac{\varphi^D}{\varphi^0}=\frac{S^0_-}{S^D_-}$ is nondecreasing by Lemma~\ref{9.9.2013.5}. 
By $\frac{\varphi^D_+}{\varphi^0_+}=\frac{S^0}{S^D}$, one obtains analogously for $\varphi^0_{u+}$ that 
\begin{align*}
0<\frac{\varphi^0_t}{\varphi^D_t}\left(x-\varphi_t^D+\varphi_{u+}^D\right)
=\frac{\varphi^0_t}{\varphi^D_t}\left(x-\varphi_t^0\frac{S_{t-}^0}{S_{t-}^D}+\varphi_{u+}^0\frac{S_{u}^0}{S_{u}^D}\right)
\leq \frac{\varphi^0_t}{\varphi^D_t} x-\varphi_t^0+\varphi_{u+}^0\quad\forall u\in(\tau_{t,x}^D,t).
\end{align*}
As $M^0_{t,x\frac{\varphi_t^0}{\varphi_t^D}}\not=\emptyset$, it can be concluded that
$\tau^0_{t,x\frac{\varphi_t^0}{\varphi_t^D}}=\sup M^0_{t,x\frac{\varphi_t^0}{\varphi_t^D}} \le \tau^D_{t,x}$.

{\em Step 3:} For $\varphi^0_t>0$ (implying $S^0_{t-}>0$ and $\varphi^D_t>0$), we have that
\beam\label{aaa}
 F^D(t,x) & = & S^D_t-\inf_{\tau_{t,x}^D\leq u\leq t} S_u^D\nonumber\\
& \stackrel{\mbox{\small Lemma~\ref{9.9.2013.5}}}{\le} & \left( \frac{S_{t-}^D}{S_{t-}^0}S_{t}^0
-\inf_{\tau_{t,x}^D\leq u< t} \frac{S_u^D}{S_u^0}S_{u}^0\right)\vee 0\nonumber\\
& \stackrel{\mbox{\small Lemma~\ref{9.9.2013.5}}}{\le} & \frac{S_{t-}^D}{S_{t-}^0}\left(S_t^0
 -\inf_{\tau_{t,x}^D\leq u< t}S_u^0\right)\vee 0\nonumber\\
& \stackrel{\mbox{\small (\ref{9.9.2013.4}) }}{\le} & \frac{S_{t-}^D}{S_{t-}^0} \left(S_t^0
-\inf_{\tau_{t,{\varphi_t^0}/{\varphi_t^D}x}\leq u< t}S_u^0\right)\vee 0\nonumber\\
& = &  \frac{\varphi_t^0}{\varphi_t^D}F^0\left(t,\frac{\varphi_t^0}{\varphi_t^D}x\right).
\eeam
Observe that for the second inequality, we use that $S_{t-}^D/S_{t-}^0\le S_u^D/S_u^0$ for $u$ {\em strictly} smaller than $t$ (all considered prices are nonzero). 
For $\varphi^0_t>0$, it follows from (\ref{aaa}) that  
\beam\label{bbb}
\alpha\int_0^\infty F^0(t,x)dx-\alpha\int_0^\infty F^D(t,x)dx
\geq \alpha\int_0^\infty F^0(t,x)dx-\alpha\int_0^\infty \frac{\varphi_t^0}{\varphi_t^D}F^0(t,x\frac{\varphi_t^0}{\varphi_t^D})dx= 0.
\eeam
If $\varphi^0_t=0$, then either $\varphi^D_t=0$ or $S^D_{t-}=0$. Both equalities imply that $F^D(t,\cdot)=0$, and, consequently, the first difference in (\ref{bbb}) is nonnegative. 
Putting (\ref{9.9.2013.3}) and (\ref{bbb}) together yields the assertion.
\end{proof}

\begin{proof}[Proof of Theorem \ref{propV}]
Let $\varphi^D\in\bbl$. $\varphi^0$  is defined as in Lemma~\ref{9.9.2013.2} and $X^D,X^0$ are the unique positions in the bank account to meet the self-financing condition 
(cf. Remark~\ref{10.9.2013.1}).
Let us first examine the right limits $V^0_+$ and $V^D_+$. By the self-financing condition, one has
\begin{align*}
V_+^0=&v_0+(1-\alpha)X^0\mal B + \varphi^0\mal S^0 - \Pi_+^0\\
V_+^D=&v_0+(1-\alpha)X^D\mal B + \varphi^D\mal S^D + \varphi^D\mal D - \Pi_+^D.
\end{align*}
On the other hand, 
\beao
V^D_+ = X^D_+ + \varphi^D_+ S^D = X^D_+ + \varphi^0_+ S^0 = V^0_+ + X^D_+ - X^0_+
\eeao
Together with $\varphi^0\mal S = \varphi^D\mal (S_t^D+D)$, one arrives at
\beao
X^D_+ - X^0_+ = V^D_+ -V^0_+ = (1-\alpha)(X^D_+-X^0_+)\mal B - \Pi^D_+ + \Pi^0_+ \le (1-\alpha)(X^D-X^0)\mal B
\eeao
By Gronwall's lemma in the form of Lemma~2.1 in \cite{katzenberger.1991} applied to the nonnegative c\`adl\`ag process $(X^D_+ - X^0_+)\vee 0$ and the nondecreasing process~$B$ (here, 
one needs that $r\ge 0$), one obtains $X^D_+ \le X^0_+$ and thus 
\beam\label{11.9.2013.1}
V^D_+ \le V^0_+.
\eeam
Note that the lemma cannot be applied directly to $X^D$ and $X^0$ as these processes are not c\`adl\`ag. Thus, the right jumps of $V^D-V^0$ have to be analyzed. 
For $\varphi^D_t=0$, one also has $\varphi^0_t=0$, and the jump at time~$t$ vanishes. Otherwise, one argues that  

{\allowdisplaybreaks\begin{eqnarray}\label{11.9.2013.2}
\Delta^+ (V^D-V^0)_t	&=& \Delta^+ (\Pi^0 - \Pi^D)_t \nonumber\\
&\stackrel{(\ref{jumpPi+})}{=}   &\alpha\int_0^{\left(\Delta^+\varphi_t^0\right)^-}F^0(t,x)dx-\alpha\int_0^{\left(\Delta^+\varphi_t^D\right)^-}F^D(t,x)dx\nonumber\\
&\stackrel{(\ref{aaa})}{\geq}&\alpha\int_0^{\left(\Delta^+\varphi_t^0\right)^-}F^0(t,x)dx
-\alpha\int_0^{\left(\Delta^+\varphi_t^D\right)^-}\frac{\varphi_t^0}{\varphi_t^D}F^0\left(t,\frac{\varphi_t^0}{\varphi_t^D}x\right)dx\nonumber\\
&\stackrel{}{=}&\alpha\int_0^{\left(\Delta^+\varphi_t^0\right)^-}F^0(t,x)dx-\alpha\int_0^{\left(\frac{\varphi_t^0}{\varphi_t^D}\Delta^+\varphi_t^D\right)^-}F^0\left(t,x\right)dx\nonumber\\
&\stackrel{}{=}&\alpha\int_0^{\left(\Delta^+\varphi_t^0\right)^-}F^0(t,x)dx-\alpha\int_0^{\left(\frac{S_{t-}^D}{S_{t-}^0}\left(\varphi_{t+}^0\frac{S_{t}^0}{S_{t}^D}
-\varphi_t^0\frac{S_{t-}^0}{S_{t-}^D}\right)\right)^-}F^0\left(t,x\right)dx\nonumber\\
&\stackrel{}{\geq}&\alpha\int_0^{\left(\Delta^+\varphi_t^0\right)^-}F^0(t,x)dx-\alpha\int_0^{\left(\frac{S_{t-}^D}{S_{t-}^0}\left(\varphi_{t+}^0\frac{S_{t-}^0}{S_{t-}^D}
-\varphi_t^0\frac{S_{t-}^0}{S_{t-}^D}\right)\right)^-}F^0\left(t,x\right)dx\nonumber\\
&\stackrel{}{=}&0
\end{eqnarray}}
The last inequality uses that $S^0_t/S^D_t\ge S^0_{t-}/S^D_{t-}$ by Lemma~\ref{9.9.2013.5}. 
Putting (\ref{11.9.2013.1}) and (\ref{11.9.2013.2}) together, one obtains
$$V_t^D=V_{t+}^D-\Delta^+ V_{t}^D\leq V_{t+}^0-\Delta^+ V_{t}^D\leq V_{t+}^0-\Delta^+ V_{t}^0=V_t^0.$$ 
\end{proof}

\section{Tax-efficient strategies}\label{9.10.2014.1} 

Let $S\ge 0$ be a continuous semimartingale and $\varphi_t=g(S_t)$ for all $t>0$, where $g:\bbr_+\to\bbr_+$ is a {\em nondecreasing} and twice continuously differentiable function.  
This means that the ``initial'' position is $\vp_{0+}=g(S_0)$, and the investor increases (reduces) her position after an increase (decrease) of the stock price. 
Denote by $g^{-1}$ the right-continuous inverse of $g$, i.e.,
\beao
g^{-1}(y) := \sup\{ s\ |\ g(s)\le y\}.
\eeao

Let us show that the book profit function reads
\beam\label{22.12.2013.1}
F(t,x) := S_t - \inf_{\tau_{t,x}\le u\le t} S_u =\left\{
\begin{array}{cl}
S_t - g^{-1}(\varphi_t-x),  & x\le \varphi_t - \inf_{0< u\le t}\varphi_u\\
S_t - \inf_{0\le u\le t}S_u, & x> \varphi_t - \inf_{0< u\le t}\varphi_u
\end{array}
\right.
\quad\mbox{for all}\ t>0,
 \eeam
which means that the infinite-dimensional stochastic process~$F$ is a direct function of the two-dimensional stochastic process~$(S_t,\inf_{0\le u\le t}S_u)_{t\ge 0}$.
Note that $\inf_{0< u\le t}\varphi_u = g(\inf_{0\le u\le t}S_u)$.

To prove (\ref{22.12.2013.1}), first consider the case that $x\le \varphi_t - \inf_{0< u\le t}\varphi_u$. By definition of $\tau_{t,x}$, one has that
$g(S_u)=\varphi_u > \varphi_t-x$ for all $u\in(\tau_{t,x},t]$. Together with the monotonicity and the continuity of $g$,
this implies that $S_u> g^{-1}(\varphi_t-x)$.
On the other hand, we have that $\varphi_{\tau_{t,x}+}=\varphi_t-x$ and thus, by $g(S_{\tau_{t,x}})=\varphi_{\tau_{t,x}+}$,
$S_{\tau_{t,x}}\le \sup\{ s\ |\ g(s)\le \varphi_{\tau_{t,x}+}\}=g^{-1}(\varphi_t-x)$ (the right limit is only needed for the case that $\tau_{t,x}=0$, which is possible if
$x= \varphi_t - \inf_{0< u\le t}\varphi_u$). By continuity of the paths of $S$,
we conclude that $\inf_{\tau_{t,x}\le u\le t}S_u = g^{-1}(\varphi_t-x)$.

This means, the purchasing price of the stock with label~$x$ is
$S_{\tau_{t,x}}=g^{-1}(\varphi_t-x)$, and up to time~$t$, the price does not fall below it. 
Now, let $x>\varphi_t - \inf_{0< u\le t}\varphi_u$. One has $\tau_{t,x}=0$ which yields the assertion.

If $g'<0$, one still has that $S_{\tau_{t,x}}=g^{-1}(\vp_t-x)$ (of course, with $g^{-1}$ defined appropriately), but now, 
$\inf_{\tau_{t,x}\le u\le t} S_u = g^{-1}(\sup_{\tau_{t,x}< u\le t}\vp_u)$, and the infimum can be attained anywhere between $\tau_{t,x}$ and $t$, 
which implies that $F(t,\cdot)$ cannot be a direct function of
$(S_t,\inf_{0\le u\le t}S_u)$.

>From (\ref{22.12.2013.1}), it follows that
\beam\label{22.12.2013.2}
\int_0^{\varphi_t} F(t,x)\,dx & = & (\varphi_t  -\inf_{0< u\le t}\varphi_u)S_t -\int_0^{\varphi_t  - \inf_{0< u\le t}\varphi_u} g^{-1}(\varphi_t-x)\,dx
+\inf_{0< u\le t}\varphi_u(S_t - \inf_{0\le u\le t}S_u)\nonumber\\
& = & \varphi_t S_t -\inf_{0< u\le t}\varphi_u\inf_{0\le u\le t}S_u  - \int_{\inf_{0< u\le t}\varphi_u}^{\varphi_t}g^{-1}(x)\,dx.
\eeam
Using that $g'=0$ on $(S_u,g^{-1}(\varphi_u))$, integration by parts yields
\beam\label{22.12.2013.3}
\int_{\inf_{0< u\le t}\varphi_u}^{\varphi_t}g^{-1}(x)\,dx & = & \int_{g^{-1}(\inf_{0< u\le t}\varphi_u)}^{g^{-1}(\varphi_t)}yg'(y)\,dy\nonumber\\
& = & \int_{\inf_{0\le u\le t}S_u}^{S_t}yg'(y)\,dy\nonumber\\
& = & y g(y) \Big|_{\inf_{0\le u\le t}S_u}^{S_t} - \int_{\inf_{0\le u\le t}S_u}^{S_t}g(y)\,dy.
\eeam

Let $G$ be an antiderivative of $g$, i.e., $G'=g$.  Putting (\ref{22.12.2013.2}) and (\ref{22.12.2013.3}) together, we arrive at
\beao
\int_0^{\varphi_t} F(t,x)\,dx = G(S_t)-G\left(\inf_{0\le u\le t} S_u\right).
\eeao

For the trading gains, one has by It\^o's formula
\beam\label{10.10.2014.1}
g(S)\mal S_t = G(S_t) - G(S_0) - \frac12 g'(S)\mal [S,S]_t = G(S_t) - G(S_0) - \frac12 [g(S),S]_t,
\eeam
which yields
\beao
\Pi_t =  \alpha \int_0^t \varphi\,dS - \alpha \int_0^{\varphi_t} F(t,x)\,dx = \alpha \left( \underbrace{G\left(\inf_{0\le u\le t}S_u\right)}_{\mbox{\small nonincreasing in $t$}} - G(S_0)
- \frac12 \underbrace{[\varphi,S]_t}_{\mbox{\small nondecreasing in\ }t}\right).
\eeao

\begin{Bemerkung}
First note that all tax payments are nonpositive (of course, only up to the liquidation of the portfolio). This is because trading gains are never realized if $g'\ge 0$.
There are two components: payments triggered by wash sales when the stock price 
reaches its running infimum~$\inf_{0\le u\le t}S_t$, and there are all the time the taxes~$-0.5\alpha [\vp,S] = -0.5\alpha g'(S)\mal [S,S]$ triggered by loss realizations 
from ``recently'' purchased stocks.  

To explain this phenomenon, consider an approximating sequence of Cox-Ross-Rubinstein type models with finite price grids~$\{0,\sigma/\sqrt{n},2\sigma/\sqrt{n},\ldots\}$, $n\in\bbn$, and 
$S^n_{(k+1)/n} -S^n_{k/n} = \pm \sigma/\sqrt{n}$ each with probability~$1/2$. First, we look at the case that at time $k/n$ the stock price lies strictly above its minimum up to this time. 
Then, the investor holds exactly $g(S^n_{k/n})-g(S^n_{k/n} -\sigma/\sqrt{n})$ shares with book profit zero. 
Namely, these shares were purchased after the last time $\le k/n$ at which $S^n$ jumps from $S^n_{k/n}-\sigma/\sqrt{n}$ to $S^n_{k/n}$. 
All other shares which are in the portfolio at time $k/n$ were purchased earlier and have a higher book profit that cannot fall strictly below zero in the next period.
Therefore, the tax payment at time $(k+1)/n$ is given by
\beao
-\alpha\left(g(S^n_{k/n})-g\left(S^n_{k/n}-\frac{\sigma}{\sqrt{n}}\right)\right)(S_{(k+1)/n} - S_{k/n})^- \approx -\alpha \frac{g'(S^n_{k/n})\sigma^2}{n} 1_{\{S_{(k+1)/n} - S_{k/n}<0\}},
\eeao
i.e., if the price goes up, there are no tax payments, and if it goes down the shares that have zero book profit before are sold.
For $n\to \infty$, by the law of large numbers, half of the price movements go down, and one arrives at the accumulated tax payments $-0.5 \alpha \int g'(S_t)\sigma^2\,dt$
(note  that in the limit the fraction of periods at which the stock price attains its running minimum vanishes). 
Then, the general case with nonconstant $d[S,S]_t/dt$ follows by stochastic time changes applied to the approximating price processes. 
If $S^n_{k/n}=\min_{l\le k}S^n_{l/n}$, all shares have book profit zero and after a further
decrease they are wash-sold, which leads to the tax payment  $\alpha g(\min_{l\le k}S^n_{l/n})(\min_{l\le k+1}S^n_{l/n}-\min_{l\le k}S^n_{l/n})$. 
In the limit, the accumulated tax payments when the stock price coincides with its running minimum become $\alpha\left(G(\inf_{0\le u\le \cdot}S_u) - G(S_0)\right)$, where $G'=g$.
\end{Bemerkung}

In general, when building up a portfolio, an investor can generate negative tax payments, or at least off-set positive tax payments on dividends, by purchasing many new stocks and 
sell whose stocks which go down. This is accompanied with higher book profits of the shares that go up. 
Thus, as time goes by, it gets increasingly more difficult to avoid tax payments.

\section{Counterexamples}\label{2.8.2013.1}

In this section, we give examples that illustrate the problems with the construction of the tax payment process and show the necessity of some assumptions.


\begin{thmB}\label{non_semimartingale}
If the stock price process is not a semimartingale, different sequences of up-approximating elementary strategies of a left-continuous strategy~$\varphi$ can lead to different 
limits of the  actual tax payments~$\Pi^n$. Namely, if $S$ is not a semimartingale, there exists a sequence of nonnegative elementary strategies~$(\varphi^n)_{n\in\mathbb{N}}$ s.t.
\begin{eqnarray*}
||\varphi^n||_{\infty}\to 0,\quad E( 1 \wedge \sup_{t\in[0,T]} (\varphi^n\mal S_t)^- )\to 0,\quad n\to \infty, 
\end{eqnarray*}
but
\begin{eqnarray*}
E( 1 \wedge \sup_{t\in[0,T]} (\varphi^n\mal S_t)^+ )\not\to 0,\quad n\to \infty, 
\end{eqnarray*}
see Theorem~1.7 of \cite{bsv2011} (shifting the strategies by the constants $||\varphi^n||_\infty$ shows that they can be chosen nonnegative). 
By $||\varphi^n||_{\infty}\to 0$, the book profits vanish, i.e., $\int_0^\infty F^n(\cdot,x)\,dx \to 0$ uniformly in probability, but the trading gains do not tend to zero.  
Thus, by Proposition~\ref{prop1}, $(\Pi^n)_{n\in\bbn}$ does not tend to zero. On the other hand, the elementary strategy $\varphi=0$, the uniform limit of $(\vp^n)_{n\in\bbn}$, 
leads to zero tax payments. 
\end{thmB}

\begin{Bemerkung}
Tax payments are not continuous w.r.t. pointwise convergence of elementary strategies. Indeed, let $\varphi^n=1_{(0, 1/2]\cup (1/2+1/n,1]}$. $\varphi^n$ converges pointwise 
to $\varphi=1_{(0,1]}$ and $\varphi^n\mal S\to \varphi\mal S$ uniformly in probability. But, in contrast to $\varphi$, the strategy $\varphi^n$ realizes 
current book profits at time $1/2$. Thus, it is not possible to define the tax payment process 
as unique continuous extension w.r.t. pointwise convergence to the space of all predictable
locally bounded strategies as it is done for the stochastic integral, cf. Theorem~I.4.31 in \cite{jacod1}. 
It seems that the convergence ``uniformly in probability'' for trading strategies is taylor-made for modeling capital gains taxes. 
The strategy set $\bbl$ is still rich enough to cover almost all relevant strategies in applications. 
\end{Bemerkung}

\section{Conclusion}

The first purpose of this paper is to find a suitable set of continuous time trading strategies (specifying the number of identical shares that an investor holds in her portfolio) 
for which the payment flow of a linear tax on realized trading gains can be constructed. 
It turns out that this is the set of all adapted processes with left-continuous paths possessing finite right limits, i.e., the closure of elementary predictable processes 
w.r.t. the convergence ``uniformly in probability''. Then, the extension to trading strategies in different stocks is straightforward.
>From a theoretical point of view, it is appealing that tax payments can also be defined for strategies of infinite variation.
This is not obvious at all because a reduction of the stock position leads to tax payments whereas an increase has no immediate effect. This property may suggest that 
a construction of the tax payment flow must be based on a decomposition into an increasing and a decreasing part of the investment strategy.
 
In the discrete time model of \citet{dyb1}, we prove that it is optimal to realize 
trading losses immediately and, when the total number of stocks has to be reduced, to sell shares with lower book profits / later purchasing times first. 
Based on this result, for elementary strategies in a continuous time model, we introduce an automatic loss realization when shares fall below their (individual) 
purchasing prices as well as a rule that dictates to sell shares with later purchasing time first when the stock position has to be reduced. 
Following this procedure, the tax payment flow is already determined by the stochastic process modeling the total number of shares in the portfolio. 
For the extension to nonelementary strategies, the representation of the book profits of the shares in the portfolio plays a key role 
(although all shares have the same price, 
their book profits differ because of different purchasing times).

Secondly, we prove that under the condition that the dividend policy has a neutral effect on the stochastic return process, 
for every investment strategy in a firm with dividends, there exists a strategy investing in an ``identical'' firm without dividends
that leads to an almost surely higher or equal after-tax wealth.

Finally, we find out tax-efficient dynamic strategies. These try to defer tax payments as long as possible. Because profit-taking leads to early tax payments, a tax-efficient strategy 
reduces the position only after losses, i.e., there should be a positive dependence between the number of stocks in the portfolio and the stock price. 
If the position is a direct {\em nondecreasing} function of the stock price, the tax payment flow can be determined explicitly and is given, besides a local time component, by 
the tax rate times half the quadratic covariation of the strategy and the price process. 


In the paper, we consider the so-called {\em exact tax basis} 
which is economically the most reasonable one. 
For other tax bases, as the FIFO (``first-in-first-out'') or the average of the purchasing prices,
the main phenomena are similar, as, e.g., the suboptimality of dividends. But, the 
modeling is quite different. Especially, it is an open problem how to construct tax payment flows 
beyond strategies of finite variation. 


\appendix

\section{Appendix: The discrete time model of Dybvig/Koo}\label{25.9.2014.1}

In this section, we motivate the automatic loss realization as well as the rule to sell shares with lower book profits / shorter residence times first
(based on this procedure, the tax payment flow was introduced for continuous time portfolio rebalancings in Section~\ref{main_results}).
For this, we prove that in the discrete time model of Dybvig and Koo~\cite{dyb1}, 
this procedure leads for all paths to a higher or equal after-tax wealth than any other strategy 
(with the same total number of shares in the portfolio) if the riskless interest rate is nonnegative. Namely, the procedure minimizes the accumulated tax payments up to any time~$t$
(see Theorem~\ref{theorem:main111}). A similar assertion is already stated in \cite{dyb1} (see Properties~1 and 2 on page~6), but in less formal terms and, so far, a proof is only 
available for Property~1 in special cases (see Subsection~3.1 of Constantinides~\cite{constantinides.1983}).
The idea is that investors always prefer tax payment obligations in the future to tax payments today. 

Following the notation in \cite{dyb1}, $N_{s,t}$ denotes the number of stocks that are bought at time $s\in\{0,\ldots,T\}$, $T\in\bbn$, and kept in the portfolio at least after trading 
at time $t\in\{s,\ldots,T\}$. Especially, $N_{t,t}$ is the number of shares purchased at time~$t$, i.e., a position cannot be purchased and resold at the same time (on the 
other hand, a position can be sold and rebought at the same time). One has the constraint
\beam\label{18.9.2014.3}
 N_{t,t}\geq N_{t,t+1}\geq\ldots\geq N_{t,T}\geq0,~\mbox{for all}~t\in\{0,\ldots,T\},
\eeam
which contains a short-selling restriction. Following the standard notation in discrete time, we denote by
\beam\label{18.9.2014.2}
\varphi_{t+1}= \sum_{s=0}^t N_{s,t},\quad t=0,\ldots,T
\eeam
the number of stocks in the portfolio {\em after} trading at time~$t$. 
Accumulated tax payments up to time~$u$ are given by 
\beam\label{18.9.2014.1}
\Pi_u:= \alpha \sum_{t=1}^u \sum_{s=0}^{t-1}\left( N_{s,t-1}-N_{s,t}\right)\left(S_t-S_s\right),
\eeam
where $\sum_{t=u+1}^{u}\ldots =0$ throughout the section.
With $\Pi$ from (\ref{18.9.2014.1}), the self-financing condition is defined as in (\ref{4.8.2013.2}).

Of course, there are different strategies~$N=(N_{s,t})_{s=0,1,\ldots,T,\ t=s,s+1,\ldots,T}$ that lead to the same number~$\varphi$ of risky assets. 
Given some nonnegative process $\vp$, the rule of selling shares on which our model in Section~\ref{main_results} is based corresponds to the following strategy~$\wt{N}$, 
constructed by (forward) induction in $t$: $\wt{N}_{0,0}=\varphi_1$ and, given $\wt{N}_{s,t-1}$, $s=0,1,\ldots,t-1$, $\wt{N}_{s,t}$ is defined as 
\begin{align}
\wt{N}_{s,t}=&1_{\left\{S_t\geq S_s\right\}}\left(\wt{N}_{s,t-1}-\left(\left(\Delta\varphi_{t+1}\right)^-
-\sum_{j=s+1}^{t-1}\wt{N}_{j,{t-1}}\right)^+\right)^+,\ s\in\{0,\ldots,t-1\},\label{WS_strategy}\\
\wt{N}_{t,t}=& \Delta \vp_{t+1} + \sum_{s=0}^{t-1}(\wt{N}_{s,{t-1}}-\wt{N}_{s,t}),\label{WS_strategy_self}
\end{align}
where $\vp_{t+1}=\Delta \vp_{t+1}-\vp_t$.
Following (\ref{WS_strategy}), the investor first reduces her total position by $(\Delta \vp_{t+1})^-$, thereby selling the shares with the smallest residence time~$t-s$.
Then, remaining shares with negative book profits are sold. By (\ref{WS_strategy_self}), condition (\ref{18.9.2014.2}) is satisfied, and, by omitting the indicator functions in 
(\ref{WS_strategy}), one sees that $\wt{N}_{t,t}\ge 0$. Now, we can already formulate the main assertion of this section. In Subsection~\ref{3.10.2014.1}, 
the precise relation to the model introduced in Section~\ref{main_results} is established. 
\begin{Satz}\label{theorem:main111}
Let $(\varphi_t)_{t\in\{1,\ldots,T+1\}}\ge 0$ be a given position in the risky asset. Let $\wt{N}$ be the strategy defined in (\ref{WS_strategy})/(\ref{WS_strategy_self})
and $N$ be an arbitrary strategy satisfying (\ref{18.9.2014.3})/(\ref{18.9.2014.2}). Then, for the corresponding accumutated tax payments, one has that 

\begin{align}\label{TP_majorant}
 \wt{\Pi}_t \leq \Pi_t\quad\mbox{for all}~t\in\{0,\ldots,T\}.
\end{align}
\end{Satz}
 
>From Theorem~\ref{theorem:main111}, it follows, as in Section~\ref{section_comparison}, that the wealth process of $\wt{N}$ dominates the wealth process of $N$ 
if the riskless interest rate 
is nonnegative. Namely, for both strategies, trading gains before taxes are given by $\vp\mal S_T:=\sum_{u=1}^T \vp_u(S_u-S_{u-1})$, 
but $\wt{N}$ defers tax payments to a larger extent.\\

Throughout the section, for $t$ and $\omega$ fixed, $(k_0,k_1,\ldots,k_t)$ is a permutation of $(0,1,\ldots,t)$ s.t.
\beam\label{24.9.2014.2}
S_{k_0}\ge S_{k_1} \ge \ldots \ge S_{k_t}\quad\mbox{and}\quad S_{k_i}>S_t\quad\forall i<j,\quad\mbox{where\ }k_j=t. 
\eeam
Then, for an arbitrary strategy~$N$, the book profit function is defined as
\begin{align}\label{def:bpf_general}
 F(t,x):=\sum_{i=0}^t (S_t-S_{k_i})1_{\left(\sum_{l=0}^{i-1} N_{k_l,t},\sum_{l=0}^{i} N_{k_l,t}\right]}(x).
\end{align}
On $(0,\vp_{t+1}]$, $F(t,\cdot)$ is obviously nondecreasing. Note that $F(t,\cdot)$ from (\ref{def:bpf_general}) already contains the portfolio regroupings that take place at price~$S_t$,
i.e., it consists of $\vp_{t+1}$ shares (see Subsection~\ref{3.10.2014.1} for the relation to the book profit function from Section~\ref{main_results}). 
To prove Theorem~\ref{theorem:main111}, we need the following lemmas.
\begin{Lemma}\label{lem:rewrite}
For every strategy $N$ with corresponding number of stocks~$\varphi$ and book profit function~$F$, one has
\begin{align*}
 \Pi_t=\alpha\varphi\mal S_t-\alpha\int_0^{\vp_{t+1}} F(t,x)dx,\quad t=0,\ldots,T
\end{align*}
(cf. Proposition~\ref{prop1}).
\end{Lemma}
\begin{proof}
We have
\begin{align*}
 &\alpha\varphi_{t} (S_t-S_{t-1}) -\alpha\left(\int_0^{\varphi_{t+1}} F(t,x)dx-\int_0^{\varphi_{t}} F(t-1,x)dx\right)\\
=&\alpha\left(\sum_{i=0}^{t-1}N_{i,t-1}(S_t-S_{t-1})-\sum_{i=0}^{t-1} \left(N_{i,t}(S_t-S_i)-N_{i,t-1}(S_{t-1}-S_i)\right)\right)\\
=&\alpha \sum_{i=0}^{t-1} \left(N_{i,t-1}-N_{i,t}\right)(S_t-S_i)\\
=& \Pi_t - \Pi_{t-1}.
\end{align*}
\end{proof}
\begin{Lemma}\label{27.9.2014.2}
Let $\wt{F}$ be the book profit function of the strategy~$\wt{N}$ from (\ref{WS_strategy}). Then, one has
\beam\label{3.10.2014.2}
\wt{F}(t,x)=1_{\left((\Delta \vp_{t+1})^+, \vp_{t+1}\right]}(x)\left(\wt{F}(t-1,x-\Delta\varphi_{t+1})+S_t-S_{t-1}\right)\vee 0.
\eeam
\end{Lemma}

\begin{proof}
If the stock price falls below the purchasing price of a particular share, then, following (\ref{WS_strategy}), this share is definitely sold. 
Consequently, one has that
\beam\label{24.9.2014.1}
S_{s_1}\le S_{s_2}\quad\mbox{for all\ }s_1<s_2\le t-1\ \mbox{with\ }\wt{N}_{s_1,t-1}>0,
\eeam 
i.e., book profits are nondecreasing in the residence time~$t-s$. Put differently, if one only considers points $s\in\{0,\ldots,t-1\}$ with $\wt{N}_{s,t-1}>0$, 
the stock price is nondecreasing. Thus, in (\ref{WS_strategy}),
the investor sells the shares with the lowest book profits (because they have the shortest residence times), 
and the strategy $\wt{N}$ given by (\ref{WS_strategy})/(\ref{WS_strategy_self}) 
reads: $\wt{N}_{0,0}=\varphi_1$ and
\allowdisplaybreaks{\begin{align*}
\wt{N}_{k_i,t}=&0,\quad\mbox{for}~i\in\{0,\ldots,j-1\},\\
\wt{N}_{k_i,t}=&\left(\wt{N}_{k_i,t-1}-\left(\left(\Delta\varphi_{t+1}\right)^--\sum_{\stackrel{l=0}{l\not=j}}^{i-1}\wt{N}_{k_l,{t-1}}\right)^+\right)^+,\quad i\in\{j+1,
\ldots,t\},\nonumber\\
\wt{N}_{k_j,t}=\wt{N}_{t,t}=&\Delta \vp_{t+1} + \sum_{\stackrel{l=0}{l\not=j}}^{t-1}(\wt{N}_{k_l,{t-1}}-\wt{N}_{k_l,t}),
\end{align*}}
for $t\in\{1,\ldots, T\}$, where $j$ and the permutation~$(k_0,\ldots,k_t)$ are given in (\ref{24.9.2014.2}).\\

{\em Case 1:} $\left(\Delta \varphi_{t+1}\right)^-\leq\sum_{l=0}^{j-1}\wt{N}_{k_l,t-1}$ (note that this includes the case $\Delta \varphi_{t+1}>0$).

One has $\wt{N}_{k_l,t}=\wt{N}_{k_l,t-1}$ for $l\ge j+1$ and arrives at
 
\beam\label{24.9.2014.3}
     \sum_{l=0}^i \wt{N}_{k_l,t}=\left\{\begin{array}{ll} 0, & i\in\{-1,0,\ldots,j-1\}\\
         \sum_{\begin{subarray}{l}l=0\\l\not=j \end{subarray}}^{i}\wt{N}_{k_l,t-1}+\Delta \varphi_{t+1}, & i\in\{j,\ldots,t\}\\ \end{array}\right. .
\eeam

For $x\in(0,\vp_{t+1}]$, one obtains 
\beao
 \wt{F}(t,x) & = & \sum_{i=0}^t (S_t-S_{k_i})1_{\left(\sum_{l=0}^{i-1}\wt{N}_{k_l,t},\sum_{l=0}^{i}\wt{N}_{k_l,t}\right]}(x)\\
        & \stackrel{(\ref{24.9.2014.3})}{=} & \sum_{i=j+1}^t(S_{t-1}-S_{k_i} + S_t - S_{t-1})1_{\left(\sum_{\stackrel{l=0}{l\not=j}}^{i-1}\wt{N}_{k_l,t-1}+\Delta \varphi_{t+1},
	\sum_{\stackrel{l=0}{l\not=j}}^{i-1}\wt{N}_{k_l,t-1}+\Delta \varphi_{t+1}\right]}(x)\\
        & = & 1_{\left((\Delta \vp_{t+1})^+, \vp_{t+1}\right]}(x)\left(\wt{F}(t-1,x-\Delta\varphi_{t+1})+(S_t-S_{t-1})\right)\vee0,
\eeao
where the last equality holds by $\left(\Delta \varphi_{t+1}\right)^+\leq\sum_{l=0}^{j-1}\wt{N}_{k_l,t-1} + \Delta \varphi_{t+1}$,
$S_t-S_{k_i}\le 0$ for $i\le j$, and $S_t-S_{k_i}\ge 0$ for $i\ge j+1$.\\

{\em Case 2:} $\left(\Delta \varphi_{t+1}\right)^->\sum_{l=0}^{j-1}\wt{N}_{k_l,t-1}$ (i.e., after the reduction of the position by $\left(\Delta \varphi_{t+1}\right)^-$, 
there are no shares with negative book profits). Define
  $$\wh{m}:=\min\left\{i\ |\ \left(\Delta\varphi_{t+1}\right)^- \le \sum_{\stackrel{l=0}{l\not=j}}^{i}\wt{N}_{k_l,{t-1}}\right\}.$$
We have $\wh{m}\ge j+1$ and arrive at
\begin{align}\label{eq:noWS_strategy}
     \sum_{l=0}^i\wt{N}_{k_l,t}=\left\{\begin{array}{ll} 0, & i\in \{-1,0,\ldots,\wh{m}-1\}\\
         \sum_{\stackrel{l=0}{l\not=j}}^{i}\wt{N}_{k_l,t-1}+\Delta \varphi_{t+1}, &
     i\in\{\wh{m},\ldots,t\}\end{array}\right. .
\end{align}

For $x\in(0,\varphi_{t+1}]$, one obtains
\beao
 \wt{F}(t,x) &=&\sum_{i=0}^t (S_t-S_{k_i})1_{\left(\sum_{l=0}^{i-1}\wt{N}_{k_l,t},\sum_{l=0}^{i}\wt{N}_{k_l,t}\right]}(x)\\
& \stackrel{(\ref{eq:noWS_strategy})}{=}&\sum_{i=\wh{m}}^t(S_{t-1}-S_{k_i}+S_t-S_{t-1})1_{\left(\sum_{\stackrel{l=0}{l\not=j}}^{i-1}\wt{N}_{k_l,t-1}+\Delta \varphi_{t+1},\sum_{\stackrel{l=0}{l\not=j}}^{i}\wt{N}_{k_l,t-1}+\Delta \varphi_{t+1}\right]}(x)\\
       &=&\wt{F}(t-1,x-\Delta\varphi_{t+1})-(S_t-S_{t-1}),
\eeao
where the last equality holds by $x-\Delta\varphi_{t+1}= x + (\Delta \vp_{t+1})^- >\sum_{\stackrel{l=0}{l\not=j}}^{\wh{m}-1}\wt{N}_{k_l,{t-1}}$ for $x>0$.

\end{proof}

\begin{proof}[Proof of Theorem \ref{theorem:main111}]
Let $N$ and $\wt{N}$ be as in the theorem with corresponding book profit functions $F$ resp. $\wt{F}$ as defined in (\ref{def:bpf_general}). Let us first show that
\beam\label{27.9.2014.1}
F(t,x)\le 1_{\left((\Delta \vp_{t+1})^+, \vp_{t+1}\right]}(x)\left(F(t-1,x-\Delta\varphi_{t+1})+S_t-S_{t-1}\right)\vee 0,\quad x\in(0,\vp_{t+1}].
\eeam
For $x\in(0,\vp_{t+1}]$, let $i\in\{0,\ldots,t\}$ s.t. $x\in\left(\sum_{l=0}^{i-1}N_{k_l,t},\sum_{l=0}^i N_{k_l,t}\right]$. If $i\le j$ (cf. (\ref{24.9.2014.2})), 
one has $F(t,x)=S_t-S_{k_i}\le 0$, and (\ref{27.9.2014.1}) holds. Thus, it remains to consider the case $i\ge j+1$. For this, we have 
\beao
x > \sum_{l=0}^{i-1}N_{k_l,t} = N_{t,t} + \sum_{\stackrel{l=0}{l\not=j}}^{i-1}N_{k_l,t} & = & \sum_{\stackrel{l=0}{l\not=j}}^{i-1}N_{k_l,{t-1}} + N_{t,t} + 
\sum_{\stackrel{l=0}{l\not=j}}^{i-1}(N_{k_l,t}-N_{k_l,{t-1}})\\
& \stackrel{(\ref{18.9.2014.3})}{\ge} & \sum_{\stackrel{l=0}{l\not=j}}^{i-1}N_{k_l,{t-1}} + N_{t,t} + \sum_{\stackrel{l=0}{l\not=j}}^t(N_{k_l,t}-N_{k_l,{t-1}})\\
& \stackrel{(\ref{18.9.2014.2})}{=} & \sum_{\stackrel{l=0}{l\not=j}}^{i-1}N_{k_l,{t-1}} +\Delta \varphi_{t+1}.
\eeao
By monotonicity of $F$, this implies $F(t-1,x-\Delta \varphi_{t+1})\ge S_{t-1}-S_{k_i} = F(t,x) - (S_t-S_{t-1})$
and together with $x>N_{t,t}\ge (\Delta \vp_{t+1})^+$, we arrive at (\ref{27.9.2014.1}).

With Lemma~\ref{27.9.2014.2} and (\ref{27.9.2014.1}), it follows by induction in $t$ that $F(t,x)\le \wt{F}(t,x)$ ,\ for all $t=1,\ldots,T$ and 
$x\in(0,\vp_{t+1}]$ (note that $F(0,x)=\wt{F}(0,x)=0$).
By Lemma \ref{lem:rewrite}, the assertion follows.
\end{proof}

\subsection{Relation to the model from Section~\ref{main_results}}\label{3.10.2014.1}

It remains to prove that the discrete time version of our model introducted in Section~\ref{main_results} does indeed coincide with the model of Dybvig/Koo with $N=\wt{N}$.  
Let $(\vp_t)_{t=1,\ldots,T+1}$ be a discrete time predictable process, i.e., $\vp_t$ is $\mathcal{F}_{t-1}$-measurable. By $\wt{F}$, we denote the corresponding discrete time 
book profit function in the sense of (\ref{def:bpf_general}) for $N=\wt{N}$. By $F$, we denote the continuous time book profit function in the sense of (\ref{ftx1})
for the piecewise constant strategy~$\sum_{n=1}^T \varphi_n1_{(n-1,n]}\in \bbl$ and the stock price process~$S=\sum_{n=0}^T S_n1_{[n,n+1)}$. 
This is the standard embedding of a discrete time market model into a continuous time framework. 
Let us show that
\beao
\wt{F}(t,x) = F(t+,x),\quad t=0,1,\ldots,T-1.
\eeao
This means that $\wt{F}(t,\cdot)$ already contains the portfolio regroupings that take place at price~$S_t$ (note that in a discrete time model, there can only be {\em one} 
change at time~$t$, whereas in continuous time, there can be a change between $t-$ and $t$ and between $t$ and $t+$).

For the piecewise constant process~$\sum_{n=1}^T \varphi_n1_{(n-1,n]}$, the right limit of the purchasing time~(\ref{13.8.2014.1}) reads
\beao
\tau_{t+,x}=\lim_{s>t,\ s\to t}\tau_{s,x}=\max\{u\in\{0,1,\ldots,t\}\ |\ \vp_u\le \vp_{t+1} - x\},\quad x\in[0,\vp_{t+1}],
\eeao 
with the convention from Section~\ref{main_results}  that $\vp_0=0$ 
(note that the increment $\vp_{u+1}-\vp_u$ is purchased at price~$S_u$). One has the implications $\tau_{t+,x}<t\quad\Rightarrow\quad \tau_{(t-1)+,x-\Delta \vp_{t+1}} = \tau_{t+,x}$
and $\tau_{t+,x}=t\quad\Leftrightarrow\quad x\le (\Delta \vp_{t+1})^+$. This implies
\beao
S_t - \min_{\tau_{t+,x}\le u\le t} S_u & = & \left(S_t-S_{t-1}+S_{t-1}-\min_{\tau_{t+,x}\le u\le t-1} S_u\right)\vee 0\\
& = & 1_{\left((\Delta \vp_{t+1})^+, \vp_{t+1}\right]}(x)\left(S_t-S_{t-1}+S_{t-1}-\min_{\tau_{(t-1)+,x-\Delta \vp_{t+1}}\le u\le t-1} S_u\right)\vee 0
\eeao
(with $\min \emptyset := \infty$), i.e., $F(t+,x)=S_t - \min_{\tau_{t+,x}\le u\le t} S_u$ satisfies the recursion~(\ref{3.10.2014.2}), and thus, it coincides 
with $\wt{F}(t,x)$. By Lemma~\ref{lem:rewrite} and Proposition~\ref{prop1}, this implies that the tax payment process defined in (\ref{18.9.2014.1}), with $N=\wt{N}$, 
coincides with the right limit of the tax payment process from Definition~\ref{Pi_elementar}.

\subsection{Proof of Proposition~\ref{30.1.2015.1}}\label{1.2.1015.1}

{\em Step 1:} Let us first prove the assertion for the  model of Dybvig/Koo with $N=\wt{N}$ for any given nonnegative $\vp$. By Theorem~\ref{theorem:main111}, we have that 
\beam\label{30.1.2015.2}
\Pi(\wt{N})=\inf_N \Pi(N),
\eeam
where $\Pi$ is defined in (\ref{18.9.2014.1}), and the infimum is taken over all nonnegative $N$ that lead to the total number~$\vp$ of shares. 
$\Pi$ is obviously linear in $N$ which already implies  positive homogeneity of (\ref{30.1.2015.2}) in $\vp$.
In addition, for any nonnegative $N^1$ leading to total number 
$\vp^1$ and any nonnegative $N^2$ leading to total number $\vp^2$, the sum $N^1+N^2\ge 0$ leads to the total position $\vp^1+\vp^2$.  
Thus, (\ref{30.1.2015.2}) is subadditive in $\vp$.\\

To see that (\ref{30.1.2015.2}) is in general not additive in $\vp$, consider $\vp^1=1_{(0,1]}$ and $\vp^2=1_{(1,2]}$, i.e., $\vp^1+\vp^2 = 1_{(0,2]}$.
For $\vp^1+\vp^2$, the tax payments at time~$1$ are $\alpha\left(S_1-\max(S_0,S_1)\right)$, but $\alpha\left(S_1-S_0\right)$ for $\vp^1$ and
zero for $\vp^2$. This already shows the non-additivity. 
Also note that for $\wt{N}$ associated to $\vp^1+\vp^2$, we have that
$\wt{N}_{0,1}=1$ if $S_1>S_0$. I.e., there is one share that is bought at time~0 and kept in the portfolio beyond time~1. On the other hand,
one has $N_{0,1}=0$ for all $N$ that lead to the total position $\vp^1$ or $\vp^2$.\\

{\em Step 2:} By Subsection~\ref{3.10.2014.1}, the properties carry over to the continuous time tax processes from Section~\ref{main_results}.
Namely, in the continuous time setting, we first fix finitely many stopping times at which $\vp$ can change its value and replace the infima 
in (\ref{16.10.2014.1}) by corresponding infima along a finite grid. 
Then, the modified tax processes coincide with the tax processes from Step~1 for an appropriate chosen discrete time market model. As the stock price is c\`adl\`ag, the infima
along the grid points converge pathwise to their continuous time counterparts when the mesh of the grid tends to zero. Thus, subadditivity and positive homogeneity
is proven for the tax processes from (\ref{16.10.2014.1}) for elementary strategies, and by Theorem~\ref{main_theorem}, they carry over to all strategies.

\addcontentsline{toc}{section}{Literaturverzeichnis}

\bibliography{bibliography.kuehn.ulbricht}

\end{document}